\documentclass{amsart}
\usepackage{amssymb}
\usepackage{amsmath}
\usepackage{amsthm}
\usepackage{mathrsfs}

\usepackage{graphicx}
\newtheorem{theorem}{Theorem}
\newtheorem{lemma}[theorem]{Lemma}

\theoremstyle{definition}
\newtheorem{definition}[theorem]{Definition}

\newtheorem{proposition}[theorem]{Proposition}

\makeatletter

\@addtoreset{theorem}{section}
\makeatother

\makeatletter

\@addtoreset{equation}{section}
\makeatother

\begin{document}                                                 
\title[On Diffusion System on an Evolving Surface with a Boundary]{On Generalized Diffusion and Heat Systems on an Evolving Surface with a Boundary}                                 
\author[Hajime Koba]{Hajime Koba}                                
\address{Graduate School of Engineering Science, Osaka University,\\
1-3 Machikaneyamacho, Toyonaka, Osaka, 560-8531, Japan}                                  
\email{iti@sigmath.es.osaka-u.ac.jp}                                      
\date{}                                       
\begin{abstract}
We consider a diffusion process on an evolving surface with a piecewise Lipschitz-continuous boundary from an energetic point of view. We employ an energetic variational approach with both surface divergence and transport theorems to derive the generalized diffusion and heat systems on the evolving surface. Moreover, we investigate the boundary conditions for the two systems to study the conservation and energy laws of them. As an application, we make a mathematical model for a diffusion process on an evolving double bubble. Especially, this paper is devoted to deriving the representation formula for the unit outer co-normal vector to the boundary of a surface.
\end{abstract}       
\maketitle

\section{Introduction}\label{sect1}

\begin{figure}[htbp]
\begin{minipage}{0.5\hsize}
\begin{center}
\includegraphics[width=50mm]{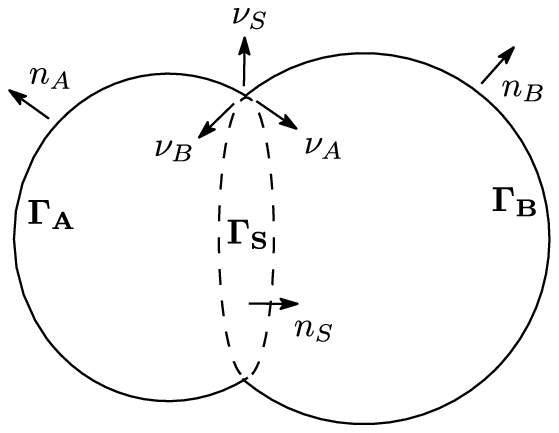}
\caption{Evolving Double Bubble}
\label{Fig1}
\end{center}
\end{minipage}%
\begin{minipage}{0.5\hsize}
\begin{center}
\includegraphics[width=50mm]{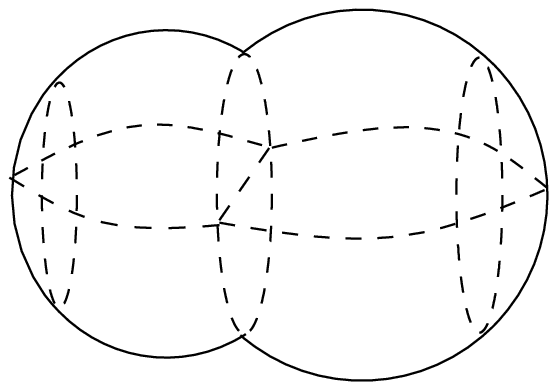}
\caption{Division of Double Bubble}
\label{Fig2}
\end{center}
\end{minipage}

\end{figure}
We are interested in diffusion processes such as concentration and thermal diffusions on an evolving double bubble (see Fig. \ref{Fig1}). This paper makes a mathematical model for a diffusion process on the evolving double bubble. To this end, we derive and study the diffusion and heat systems on an evolving surface with a boundary from an energetic point of view. An evolving surface means that the surface is moving or the shape of the surface is changing along with the time. Especially, this paper investigates the boundary conditions in co-normal direction for our systems.

Let us first introduce basic notation. Let $x = { }^t ( x_1 , x_2, x_3) \in \mathbb{R}^3$, $X = { }^t (X_1 , X_2) \in \mathbb{R}^2$ be the spatial variables, and $t , \tau \geq 0$ be the time variable. Let $T \in (0, \infty]$, and let $\Gamma (t) (= \{ \Gamma (t) \}_{0 \leq t < T}) $ be an evolving surface with a piecewise Lipschitz-continuous boundary. The symbol $n = n (x,t) = { }^t (n_1,n_2 , n_3)$ denotes the unit outer normal vector at $x \in \overline{\Gamma (t)}$, and $\nu = \nu ( x, t) = { }^t (\nu_1 , \nu_2 , \nu_3)$ denotes the unit outer co-normal vector at $x \in \partial \Gamma (t)$.  The notation $w = w (x , t) = { }^t (w_1 , w_2 ,w_3)$ means the \emph{motion velocity} of the evolving surface $\Gamma (t)$. We often call $w$ the \emph{speed} of the evolving surface $\Gamma (t)$. See Section \ref{sect2} for the definitions of $\Gamma (t)$, $\nu$, and $w$. Set
\begin{equation*}
\Gamma_T = \bigcup_{0 <  t  < T}\{ \Gamma (t) \times \{ t \} \},{ \ \ \ }\partial \Gamma_T = \bigcup_{0 <  t  < T}\{ \partial \Gamma (t) \times \{ t \} \}.
\end{equation*}
The notation $C = C ( x , t )$ represent the \emph{concentration of amount of a substance} on $\Gamma (t)$, and $\rho = \rho ( x, t )$ and $\theta = \theta (x , t)$ represent the \emph{density} and the \emph{temperature} of a substance on $\Gamma (t)$, respectively.

This paper has three purposes. The first one is to derive and study the following \emph{generalized diffusion system} on the evolving surface $\Gamma_T$:
\begin{equation}\label{eq11}
D_t^w C + ({\rm{div}}_\Gamma w ) C = {\rm{div}}_\Gamma \{ e_1' (| {\rm{grad}}_\Gamma C |^2 ) {\rm{grad}}_\Gamma C \} \text{ on } \Gamma_T,
\end{equation}
and the following \emph{generalized heat system} on the evolving surface $\Gamma_T$:
\begin{equation}\label{eq12}
\begin{cases}
D_t^w \rho + ({\rm{div}}_\Gamma w ) \rho = 0 & \text{ on } \Gamma_T,\\
\rho D_t^w \theta = {\rm{div}}_\Gamma \{ e_2' (| {\rm{grad}}_\Gamma \theta |^2 ) {\rm{grad}}_\Gamma \theta \} & \text{ on }\Gamma_T,
\end{cases}
\end{equation}
where $D_t^w f = \partial_t f + ( w \cdot \nabla ) f$, $\partial_t = \partial/{\partial t}$, $\nabla = { }^t (\partial_1 , \partial_2 , \partial_3)$, $\partial_j = \partial/{\partial x_j}$, ${\rm{div}}_\Gamma w = \nabla_\Gamma \cdot w$, ${\rm{grad}}_\Gamma f = \nabla_\Gamma f $, $\nabla_\Gamma = { }^t (\partial_1^\Gamma , \partial_2^\Gamma , \partial_3^\Gamma )$, and $\partial_j^\Gamma f = \sum_{i=1}^3 (\delta_{i j}  - n_j n_i ) \partial_i f$. Here $e_1$ and $e_2$ are two $C^1$-functions, and $\displaystyle{e'_{\mathfrak{j}} = e'_{\mathfrak{j}} ( r) = {d e_{\mathfrak{j} }}/{dr}(r)}$ ($\mathfrak{j} =1,2$). This paper derives the two systems \eqref{eq11} and \eqref{eq12} from an energetic point of view. More precisely, we use the energy density $e_1 (| {\rm{grad}}_\Gamma C |^2 )$ for the energy dissipation due to general diffusion, the energy density $e_2 (| {\rm{grad}}_\Gamma \theta |^2 )$ for the energy dissipation due to thermal diffusion, and an energetic variational approach to derive the two systems. Moreover, we investigate the boundary conditions in co-normal direction for the two systems to consider the conservation and energy laws of them. In fact, if $\partial C/\partial \nu|_{\partial \Gamma (t)} = 0$, then the system \eqref{eq11} satisfies the following conservation law$(t_1 < t_2)$:
\begin{equation}\label{eq13}
\int_{\Gamma (t_2)} C (x,t_2) { \ } d \mathcal{H}^2_x = \int_{\Gamma (t_1)} C (x,t_1) { \ } d \mathcal{H}^2_x,
\end{equation}
and the following energy law$(t_1 < t_2)$:
\begin{multline}\label{eq14}
\int_{\Gamma (t_2)} \frac{1}{2} | C (x,t_2) |^2 { \ } d \mathcal{H}^2_x + \int_{t_1}^{t_2} \int_{\Gamma (\tau)} e_1' (| {\rm{grad}}_\Gamma C |^2 ) | {\rm{grad}}_\Gamma C |^2 { \ }d \mathcal{H}^2_x d \tau \\ = \int_{\Gamma (t_1)} \frac{1}{2} | C (x,t_1) |^2 { \ } d \mathcal{H}^2_x.
\end{multline}
Here $\partial f/\partial \nu := ( \nu \cdot \nabla_\Gamma) f$. See Theorem \ref{thm27} for the conservation and energy laws of the system \eqref{eq12}. Remark that the system \eqref{eq11} is equivalent to
\begin{equation}\label{eq15}
D_t^N C + {\rm{div}}_\Gamma \{ C w - e_1' (| {\rm{grad}}_\Gamma C |^2 ) {\rm{grad}}_\Gamma C \} = 0 \text{ on }\Gamma_T,
\end{equation}
and that the system \eqref{eq12} is equivalent to
\begin{equation}\label{eq16}
\begin{cases}
D_t^N \rho + {\rm{div}}_\Gamma ( \rho w ) = 0 & \text{ on }\Gamma_T,\\
D_t^N ( \rho \theta ) + {\rm{div}}_\Gamma \{ \rho \theta w - e_2' (| {\rm{grad}}_\Gamma \theta |^2 ) {\rm{grad}}_\Gamma \theta \} = 0 & \text{ on }\Gamma_T.
\end{cases}
\end{equation}
Here $D_t^N f = \partial_t f + ( w \cdot n) (n \cdot \nabla ) f $. The forms \eqref{eq15} and \eqref{eq16} are conservative forms of the systems \eqref{eq11} and \eqref{eq12}, respectively. Remark also that
\begin{equation*}
{\rm{div}}_\Gamma \{ e_1' (| \nabla_\Gamma C |^2 ) \nabla_\Gamma C \} = 
\begin{cases}
\Delta_\Gamma C \text{ if } e_1 ( r ) = r,\\
{\rm{div}}_\Gamma \{ | {\rm{grad}}_\Gamma C |^{2 p} {\rm{grad}}_\Gamma C  \} \text{ if }e_1 (r) = \frac{ r^{p + 1}}{p+1},\\
{\rm{div}}_\Gamma \left\{ \frac{ {\rm{grad}}_\Gamma C}{1 + | {\rm{grad}}_\Gamma C |^2} \right\} \text{ if } e_1 ( r ) = \log ( 1 + r ).
\end{cases}
\end{equation*}
Here $\Delta_\Gamma$ is the \emph{Laplace-Beltrami operator} defined by $\Delta_\Gamma f = (\partial_1^\Gamma)^2 f + (\partial_2^\Gamma)^2 f + (\partial_3^\Gamma)^2 f$. See Section \ref{sect4} for our mathematical modeling of \eqref{eq11} and \eqref{eq12}.

The second one is to make a mathematical model for a diffusion process on an evolving double bubble (see Fig. \ref{Fig1}) from an energetic point of view. More precisely, we apply an energetic variational approach to derive the following diffusion system on the evolving double bubble:
\begin{equation}\label{eq17}
\begin{cases}
D_t^{w_A} C_A + ({\rm{div}}_{\Gamma_A} w_A ) C_A = {\rm{div}}_{\Gamma_A} \{ \kappa_A {\rm{grad}}_{\Gamma_A} C_A \} & \text{ on } \Gamma_A^T,\\
D_t^{w_B} C_B + ({\rm{div}}_{\Gamma_B} w_B ) C_B = {\rm{div}}_{\Gamma_B} \{ \kappa_B {\rm{grad}}_{\Gamma_B} C_B \} & \text{ on } \Gamma_B^T,\\
D_t^{w_S} C_S + ({\rm{div}}_{\Gamma_S} w_S ) C_S = {\rm{div}}_{\Gamma_S} \{ \kappa_S {\rm{grad}}_{\Gamma_S} C_S \} & \text{ on } \Gamma_S^T,\\
C_A|_{\partial \Gamma_A} = C_B|_{\partial \Gamma_B} = C_S|_{\partial \Gamma_S},\\
\kappa_A \frac{\partial C_A}{\partial \nu_A} \big|_{\partial \Gamma_A} + \kappa_B \frac{ \partial C_B}{\partial \nu_B} \big|_{\partial \Gamma_B} + \kappa_S \frac{\partial C_S}{\partial \nu_S} \big|_{\partial \Gamma_S} =0,
\end{cases}
\end{equation}
where $D_t^{w_\sharp} f = \partial_t f + (w_\sharp \cdot \nabla ) f$, ${\partial C_\sharp}/{\partial \nu_\sharp} = ( \nu_\sharp \cdot \nabla_{\Gamma_\sharp} ) C_\sharp $,
\begin{align*}
\Gamma_\sharp^T = \bigcup_{0 < t < T} \{ \Gamma_\sharp (t) \times \{ t \} \}, { \ }\partial \Gamma_\sharp^T = \bigcup_{0 < t < T} \{ \partial \Gamma_\sharp (t) \times \{ t \} \}, { \ }\sharp = A,B, S. 
\end{align*}
Here $\Gamma_A (t)$, $\Gamma_B (t)$, $\Gamma_S (t)$ are three evolving surfaces with a boundary (see Section \ref{sect5} for the assumptions of $\Gamma_A (t)$, $\Gamma_B (t)$, and $\Gamma_S (t)$), $C_\sharp = C_\sharp (x, t)$ denotes the \emph{concentration of a substance} on $\Gamma_\sharp (t)$, $\kappa_\sharp = \kappa_\sharp (x, t)$ denotes the \emph{diffusion coefficient of the substance} on $\Gamma_\sharp (t)$, $w_\sharp = w_\sharp (x,t) = { }^t ( w_1^\sharp , w_2^\sharp , w_3^\sharp )$ denotes the \emph{motion velocity} of $\Gamma_\sharp (t)$, and $\nu_\sharp = \nu_\sharp (x,t) = { }^t ( \nu_1^\sharp , \nu_2^\sharp , \nu_3^\sharp )$ denotes the unit outer co-normal vector to $\partial \Gamma_\sharp (t)$. This paper applies an energetic variational approach and both transport and divergence theorems on the evolving double bubble to derive the system \eqref{eq17}. Remark that the system \eqref{eq17} satisfies the following conservation law$(t_1 < t_2)$:
\begin{equation}\label{eq18}
\sum_{\sharp = A,B,S} \int_{\Gamma_\sharp (t_2)} C_\sharp ( x, t_2 ) { \ }  d \mathcal{H}^2_x = \sum_{\sharp = A,B,S} \int_{\Gamma_\sharp (t_1)} C_\sharp ( x , t_1) { \ }  d \mathcal{H}^2_x,
\end{equation}
and the following energy law$(t_1 < t_2)$:
\begin{multline}\label{eq19}
\sum_{\sharp = A,B,S} \int_{\Gamma_\sharp (t_2)} \frac{1}{2} |C_\sharp |^2 { \ }  d \mathcal{H}^2_x + \sum_{\sharp = A,B,S} \int_{t_1}^{t_2} \int_{\Gamma_\sharp ( \tau )} \kappa_\sharp | \nabla_{\Gamma_\sharp} C_\sharp |^2 { \ } d \mathcal{H}^2_x d \tau\\
 = \sum_{\sharp = A,B,S} \int_{\Gamma_\sharp (t_1)} \frac{1}{2} | C_\sharp |^2 { \ }  d \mathcal{H}^2_x.
\end{multline}
Here $\sum_{\sharp = A,B,S} f_\sharp = f_A + f_B + f_S$. Remark also that the system \eqref{eq17} is equivalent to
\begin{equation}\label{Eq110}
D_t^{N_\sharp} C_\sharp + {\rm{div}}_{\Gamma_\sharp} \{ C_\sharp w_\sharp - \kappa_\sharp {\rm{grad}}_{\Gamma_\sharp} C_\sharp \} = 0 \text{ on } \Gamma_\sharp^T,
\end{equation}
where $\sharp = A , B ,S$, and $D_t^{N_\sharp} f = \partial_t f + ( w_\sharp \cdot n_\sharp ) (n_\sharp \cdot \nabla_{\Gamma_\sharp} ) f $. Here $n_\sharp = n_\sharp (x,t) = { }^t (n_1^\sharp , n_2^\sharp , n_3^\sharp )$ denotes the unit outer normal vector to $\overline{\Gamma_\sharp (t)}$. See Section \ref{sect6} for our mathematical modeling of \eqref{eq17}.

The third one is to derive the following surface divergence theorem:
\begin{equation*}
\int_{\Gamma (t)} {\rm{div}}_\Gamma \varphi { \ }d \mathcal{H}^2_x = - \int_{\Gamma (t)} H_\Gamma (n \cdot \varphi ) { \ } d \mathcal{H}^2_x + \int_{\partial \Gamma (t)} \nu \cdot \varphi { \ } d \mathcal{H}^1_x
\end{equation*}
for every $\varphi = { }^t ( \varphi_1, \varphi_2 , \varphi_3 ) \in [ C^1 ( \overline{\Gamma (t)}) ]^3$. Here $H_\Gamma = H_\Gamma (x,t)$ denotes the \emph{mean curvature} in the direction $n$ and $d \mathcal{H}_x^k$ denotes the \emph{$k$-dimensional Hausdorff measure}. Especially, this paper is devoted to deriving the representation formula for the unit outer co-normal vector to the boundary of the surface $\Gamma (t)$. See Theorem \ref{thm23} for details.

Let us state three key ideas for making a mathematical model for a diffusion process on an evolving double bubble. The first one is to use an energetic variational approach. An energetic variational approach is the way to derive PDEs from an energetic point of view (see Strutt \cite{Str73} and Onsager (\cite{Ons31a}, \cite{Ons31b})). This paper focuses on the energy dissipation due to surface diffusion, and calculates a variation of the dissipation energy to derive the diffusion terms of our system. The second one is to divide the double bubble into some evolving surfaces with boundaries (see Fig. \ref{Fig2}). By dividing the double bubble into some evolving surfaces with boundaries, we can make use of surface transport and divergence theorems on an evolving surface with a boundary to derive the transport and divergence theorems on the evolving double bubble. The third one is to derive the representation formula for the unit outer co-normal vector to the boundary of a surface. Using the representation formula for the unit outer co-normal vector, we can investigate the boundary conditions for our systems.

Now we state some references on mathematical modeling of thermodynamics and fluid dynamics on an evolving surface. Dziuk-Elliott \cite{DE07} used the surface transport theorem and their flux to make their fluid system on an evolving closed surface. Koba-Liu-Giga \cite{KLG17} and Koba \cite{K18} applied their energetic variational approach to derive their fluid systems on an evolving closed surface. This paper make use of an energetic variational approach to derives the generalized diffusion and heat systems on an evolving surface with a boundary, and makes a mathematical model for a diffusion process on an evolving double bubble. Recently, Koba \cite{K19} applied their energetic variational approach to derive the compressible fluid system on an evolving surface with a boundary.

The outline of this paper is as follow: In Section \ref{sect2} we state the definition of an evolving surface with a boundary and the main results of this paper. In Section \ref{sect3} we study the surface divergence theorem and the unit outer co-normal vector to the boundary of a surface. In Section \ref{sect4} we drive the diffusion system \eqref{eq11} and the heat system \eqref{eq12} on an evolving surface with a boundary, and investigate the boundary conditions for the two systems to consider the conservation and energy laws of them. In Section \ref{sect5} we study the divergence and transport theorems on an evolving double bubble. In Section \ref{sect6} we apply an energetic variational approach to derive the diffusion system \eqref{eq17} on an evolving double bubble, and investigate the conservation and energy laws of the system.

\section{Main Results}\label{sect2}

We first introduce the definition of an evolving surface with a boundary and notation. Then, we state the main results of this paper.

\begin{definition}[Bounded domain with a boundary]\label{def21}{ \ }
Let $U \subset \mathbb{R}^2$ be a bounded domain. We call $U$ a \emph{bounded domain with a piecewise Lipschitz-continuous boundary} if the boundary $\partial U$ of $U$ can be written by
\begin{equation*}
\partial U = \partial U^1 \cup \partial U^2 \cup \cdots \cup \partial U^M \text{ for some } M \in \mathbb{N}.
\end{equation*}
Here
\begin{equation*}
\partial U^m = \{ X = { }^t ( X_1 , X_2 ) \in \mathbb{R}^2; { \ } X_1 = p_m ( \ell ), X_2 = q_m ( \ell ), { \ }\ell \in [a_m , b_m ] \},
\end{equation*}
where $m \in \{1,2, \cdots, M\}$, $a_m , b_m \in \mathbb{R}$, $p_m, q_m \in C^{0,1}([ a_m , b_m ])$, satisfying $(\mathrm{i})$ $(p_{m} (b_m ) , q_{m} ( b_m)) = (p_{m + 1} (a_{m+1} ) , q_{m+1} (a_{m+1}))$ for each $m \in \{1 ,2 , \cdots, M - 1 \}$, $(\mathrm{ii})$ $(p_{M} (b_M ) , q_{M} (b_M)) = (p_1 (a_1 ) , q_1 (a_1))$, $(\mathrm{iii})$ 
\begin{equation*}
\left( \frac{d p_m }{d \ell} ( \ell ) \right)^2 + \left( \frac{d q_m}{d \ell} ( \ell ) \right)^2 > 0 \text{ for almost all } \ell \in [ a_m , b_m ].
\end{equation*}
\end{definition}

\begin{definition}[Evolving surface with a boundary]\label{def22}
For $0 \leq t < T$, let $\Gamma (t) \subset \mathbb{R}^3$ be a set. We call $\Gamma (t) (= \{ \Gamma (t) \}_{0 \leq t <T})$ an \emph{evolving surface with a piecewise Lipschitz-continuous boundary} if there are bounded domain $U$ with a piecewise Lipschitz-continuous boundary and $\widehat{x} = { }^t ( \widehat{x}_1 , \widehat{x}_2 , \widehat{x}_3) \in [ C^2 ( \overline{U} \times [0,T ))]^3$ satisfying the properties as in Definition \ref{def21} and the following three properties hold:\\
$(\mathrm{i})$ For every $0 \leq t < T$, $\Gamma (t)$ can be written by
\begin{equation*}
\Gamma (t) = \{ x = { }^t (x_1, x_2,x_3) \in \mathbb{R}^3; { \ }x = \widehat{x}(X , t) , X \in U \}.
\end{equation*}
$(\mathrm{ii})$ For each $0 < t <T$,
\begin{equation*}
\widehat{x} ( \cdot , t) : U \to \Gamma (t) \text{ is bijective}.
\end{equation*}
$(\mathrm{iii})$ For every $0 \leq t < T$, the boundary $\partial \Gamma (t)$ of $\Gamma (t)$ can be written by
\begin{equation*}
\partial \Gamma (t) = \partial \Gamma^1 (t) \cup \partial \Gamma^2 (t) \cup \cdots \cup \partial \Gamma^M (t),
\end{equation*}
where
\begin{equation*}
\partial \Gamma^m (t) = \{ x \in \mathbb{R}^3; { \ } x = \widehat{x} (X, t ), X \in \partial U^m \}.
\end{equation*}
$(\mathrm{iv})$ For every $0 \leq t < T$ and $X \in \overline{U}$,
\begin{multline*}
\left( \frac{\partial \widehat{x}_2}{\partial X_1} \frac{\partial \widehat{x}_3}{\partial X_2} - \frac{\partial \widehat{x}_2}{\partial X_2} \frac{\partial \widehat{x}_3}{\partial X_1} \right)^2 + \left( \frac{\partial \widehat{x}_1}{\partial X_1} \frac{\partial \widehat{x}_3}{\partial X_2} - \frac{\partial \widehat{x}_1}{\partial X_2} \frac{\partial \widehat{x}_3}{\partial X_1} \right)^2\\
 + \left( \frac{\partial \widehat{x}_1}{\partial X_1} \frac{\partial \widehat{x}_2}{\partial X_2} - \frac{\partial \widehat{x}_1}{\partial X_2} \frac{\partial \widehat{x}_2}{\partial X_1} \right)^2 > 0.
\end{multline*}
Here
\begin{equation*}
C^2 ( \overline{U} \times [0,T) ) := \{ f : \overline{U} \times [0,T) \to \mathbb{R}; { \ }f = \mathcal{F}|_{\overline{U} \times [0,T)} \text{ for some } \mathcal{F} \in C^2 ( \mathbb{R}^3) \} .
\end{equation*}
\end{definition}

Let us explain the conventions used in this paper. We use the Greek characters $\alpha , \beta$ as $1.2$. Moreover, we often use the following Einstein summation convention: $c_\alpha d_{\alpha \beta} = \sum_{\alpha = 1}^2 c_\alpha d_{\alpha \beta}$ and $c_\beta d^{\alpha \beta} = \sum_{\beta = 1}^2 c_\beta d^{\alpha \beta}$.

Next we define notation. Let $\Gamma (t) (= \{ \Gamma (t) \}_{0 \leq t <T})$ be an evolving surface with a piecewise Lipschitz-continuous boundary. By definition, there are bounded domain $U \subset \mathbb{R}^2$ with a a piecewise Lipschitz-continuous boundary and $\widehat{x} = \widehat{x} (X,t) = { }^t ( \widehat{x}_1 , \widehat{x}_2 , \widehat{x}_3 ) \in [ C^2 ( \overline{U} \times [0, T ) ) ]^3$ satisfying the properties as in Definitions \ref{def21} and \ref{def22}. The symbol $n = n (x , t) = { }^t (n_1 , n_2 , n_3)$ denotes the unit outer normal vector at $x \in \overline{\Gamma (t)}$, the symbol $\nu = \nu ( x, t) = { }^t ( \nu_1 , \nu_2 , \nu_3)$ denotes the unit outer co-normal vector at $x \in \partial \Gamma (t)$, and the symbol $n^U = n^U (X) = { }^t (n_1^U , n_2^U)$ denotes the unit outer normal vector at $X \in \partial U$. Note that $\nu (x,t)$ exists for almost all $x \in \partial \Gamma (t)$ and $n^U (X)$ exists for almost all $X \in \partial U$ since $\partial \Gamma (t)$ and $\partial U$ are piecewise Lipschitz-continuous boundaries. For every $0 \leq t < T$ and $X \in \overline{U}$,
\begin{multline*}
\mathfrak{g}_1 = \mathfrak{g}_1 (X,t) := \frac{\partial \widehat{x}}{\partial X_1},{ \ }\mathfrak{g}_2 = \mathfrak{g}_2 (X,t) := \frac{\partial \widehat{x}}{\partial X_2},{ \ }\mathfrak{g}_{\alpha \beta} := \mathfrak{g}_\alpha \cdot \mathfrak{g}_\beta,\\
( \mathfrak{g}^{ \alpha \beta })_{2 \times 2} := ( \mathfrak{g}_{\alpha \beta })_{2 \times 2}^{-1} = \frac{1}{\mathfrak{g}_{11} \mathfrak{g}_{22} - \mathfrak{g}_{12} \mathfrak{g}_{21} } 
\begin{pmatrix}
\mathfrak{g}_{22} & - \mathfrak{g}_{21}\\
- \mathfrak{g}_{12} & \mathfrak{g}_{11}
\end{pmatrix},{ \ }\mathfrak{g}^\alpha := \mathfrak{g}^{\alpha \beta} \mathfrak{g}_\beta,\\
\mathcal{G} = \mathcal{G} ( X , t) := \mathfrak{g}_{11} \mathfrak{g}_{22} - \mathfrak{g}_{12} \mathfrak{g}_{21}.
\end{multline*}
Note that $\mathfrak{g}^{\alpha} \cdot \mathfrak{g}_\beta = \delta_{ \alpha \beta}$, $\mathfrak{g}^{\alpha \beta} = \mathfrak{g}^\alpha \cdot \mathfrak{g}^\beta$,
\begin{multline*}
\mathcal{G} = | \mathfrak{g}_1 \times \mathfrak{g}_2 |^2 =\mathfrak{g}_{11} \mathfrak{g}_{22} - \mathfrak{g}_{12} \mathfrak{g}_{21} = \left( \frac{\partial \widehat{x}_2}{\partial X_1} \frac{\partial \widehat{x}_3}{\partial X_2} - \frac{\partial \widehat{x}_2}{\partial X_2} \frac{\partial \widehat{x}_3}{\partial X_1} \right)^2\\
 + \left( \frac{\partial \widehat{x}_1}{\partial X_1} \frac{\partial \widehat{x}_3}{\partial X_2} - \frac{\partial \widehat{x}_1}{\partial X_2} \frac{\partial \widehat{x}_3}{\partial X_1} \right)^2   + \left( \frac{\partial \widehat{x}_1}{\partial X_1} \frac{\partial \widehat{x}_2}{\partial X_2} - \frac{\partial \widehat{x}_1}{\partial X_2} \frac{\partial \widehat{x}_2}{\partial X_1} \right)^2 > 0.
\end{multline*}
Here $\delta_{\alpha \beta}$ denotes the Kronecker delta. Note also that
\begin{equation*}
n ( \widehat{x} (X , t) , t) = \frac{\mathfrak{g}_1 \times \mathfrak{g}_2}{ |\mathfrak{g}_1 \times \mathfrak{g}_2| } = 
\frac{1}{ \sqrt{ \mathfrak{g}_{11} \mathfrak{g}_{22} - \mathfrak{g}_{12} \mathfrak{g}_{21} } }
\begin{pmatrix}
\frac{\partial \widehat{x}_2}{\partial X_1} \frac{\partial \widehat{x}_3}{\partial X_2} - \frac{\partial \widehat{x}_2}{\partial X_2} \frac{\partial \widehat{x}_3}{\partial X_1}\\
\frac{\partial \widehat{x}_1}{\partial X_1} \frac{\partial \widehat{x}_3}{\partial X_2} - \frac{\partial \widehat{x}_1}{\partial X_2} \frac{\partial \widehat{x}_3}{\partial X_1} \\
\frac{\partial \widehat{x}_1}{\partial X_1} \frac{\partial \widehat{x}_2}{\partial X_2} - \frac{\partial \widehat{x}_1}{\partial X_2} \frac{\partial \widehat{x}_2}{\partial X_1}
\end{pmatrix}.
\end{equation*}
See Jost \cite{Jos11} and Ciarlet \cite{Cia05} for differential geometry and the Riemannian manifold.

Let us introduce function spaces and notation. Set
\begin{align*}
\Gamma_T & = \bigcup_{0 <  t  < T}\{ \Gamma (t) \times \{ t \} \},{ \ }\partial \Gamma_T = \bigcup_{0 <  t  < T}\{ \partial \Gamma (t) \times \{ t \} \},\\
\overline{\Gamma_T} & = \bigcup_{0 \leq  t  < T}\{ \overline{\Gamma (t)} \times \{ t \} \},
\end{align*}
and for $k \in \mathbb{N} \cup \{ 0 , \infty \}$,
\begin{align*}
C^k ( \Gamma (t)) & := \{ f : \Gamma (t) \to \mathbb{R}; { \ }f = \mathcal{F}|_{\Gamma (t)} \text{ for some } \mathcal{F} \in C^k ( \mathbb{R}^3) \},\\
C^k ( \overline{\Gamma (t)}) & := \{ f : \overline{\Gamma (t)} \to \mathbb{R}; { \ }f = \mathcal{F}|_{\overline{\Gamma (t)}} \text{ for some } \mathcal{F} \in C^k ( \mathbb{R}^3) \},\\
C_0^k (\Gamma (t)) & := \{ f \in C^k (\Gamma (t)); { \ }\text{supp} f \text{ does not intersect } \partial \Gamma (t) \},\\
C^k ( \overline{ \Gamma_T} ) & := \{ f : \overline{\Gamma_T} \to \mathbb{R}; { \ } f = \mathcal{F}|_{\overline{\Gamma_T}} \text{ for some } \mathcal{F} \in C^k ( \mathbb{R}^4) \},\\
C^k ( \partial \Gamma (t) ) & := \{ f : \partial \Gamma (t) \to \mathbb{R}; { \ } f = \mathcal{F}|_{\partial \Gamma (t)} \text{ for some } \mathcal{F} \in C^k ( \mathbb{R}^3) \},\\
C^k ( \partial \Gamma_T ) & := \{ f :\partial \Gamma_T \to \mathbb{R}; { \ } f = \mathcal{F}|_{\partial \Gamma_T} \text{ for some } \mathcal{F} \in C^k ( \mathbb{R}^4) \}.
\end{align*}
Fix $j \in \{ 1 , 2 , 3 \}$. For $f \in C^1 ( \overline{\Gamma (t)} )$ and $\varphi = { }^t ( \varphi_1, \varphi_2 , \varphi_3) \in [C^1 ( \overline{ \Gamma (t) })]^3$,
\begin{align*}
\partial_j^\Gamma f & := \sum_{i=1}^3 (\delta_{i j} - n_i n_j ) \partial_i f,\\ 
\nabla_\Gamma & :=  { }^t ( \partial_1^\Gamma , \partial_2^\Gamma , \partial_3^\Gamma ),\\
{\rm{grad}}_\Gamma f & := \nabla_\Gamma f = { }^t (\partial_1^\Gamma f, \partial_2^\Gamma f, \partial_3^\Gamma f) ,\\
{\rm{div}}_\Gamma \varphi & := \nabla_\Gamma \cdot \varphi = \partial_1^\Gamma \varphi_1 + \partial_2^\Gamma \varphi_2 + \partial_3^\Gamma \varphi_3,\\
\frac{\partial f}{\partial \nu} & := ( \nu \cdot \nabla_\Gamma ) f . 
\end{align*}
Here $\delta_{ij}$ is the Kronecker delta. The symbol $H_\Gamma = H_\Gamma (x , t )$ denotes the mean curvature in the direction $n$ defined by
\begin{equation*}
H_\Gamma = - {\rm{div}}_\Gamma n .
\end{equation*}
Throughout this paper we assume that there is $w = w (x, t) = { }^t ( w_1 , w_2 , w_3) \in C ( \overline{\Gamma}_T)$ such that for $0 < t < T$ and $X  \in \overline{U}$,
\begin{equation*}
w ( \widehat{x} (X,t) , t) = \frac{\partial \widehat{x} }{\partial t} (X,t) .
\end{equation*}
We call $w$ the \emph{motion velocity} of the evolving surface $\Gamma (t)$.

Now we state the main results of this paper.
\begin{theorem}[Surface divergence theorem]\label{thm23}{ \ }\\
For every $\varphi = { }^t ( \varphi_1 , \varphi_2 , \varphi_3 ) \in C^1 ( \overline{ \Gamma (t)})$,
\begin{equation}\label{eq21}
\int_{\Gamma (t)} {\rm{div}}_\Gamma \varphi { \ }d \mathcal{H}^2_x = - \int_{\Gamma (t)} H_\Gamma (n \cdot \varphi) { \ } d \mathcal{H}^2_x + \int_{\partial \Gamma (t)} \nu \cdot \varphi { \ } d \mathcal{H}^1_x.
\end{equation}
Here the unit outer co-normal vector $\nu = \nu (x,t) = { }^t ( \nu_1 , \nu_2 , \nu_3)$ is given by
\begin{equation}\label{eq22}
\nu = \nu ( \widehat{x} (X,t) , t ) = \frac{n_1^U \mathfrak{g}_2 - n_2^U \mathfrak{g}_1}{ | n_1^U \mathfrak{g}_2 -n_2^U \mathfrak{g}_1| } \times \frac{\mathfrak{g}_1 \times \mathfrak{g}_2}{ | \mathfrak{g}_1 \times \mathfrak{g}_2 | } \text{ for a.e. } X \in \partial U,
\end{equation}
where $n^U = { }^t (n^U_1, n^U_2)$ is the unit outer normal vector to $\partial U$ and $\mathfrak{g}_\alpha = {\partial \widehat{x}}/{\partial X_\alpha}$. Note that $\nu \cdot \nu =1$ and $\nu \cdot n = 0$ on $\partial \Gamma (t)$.
\end{theorem}
\noindent Simon \cite{Sim83} derived \eqref{eq21} without \eqref{eq22} in the case when $\Gamma (t)$ is a fixed compact $C^2$-manifold with a smooth boundary. This paper deals with surfaces with a piecewise Lipschitz-continuous boundary and gives another proof of Theorem \ref{thm23}. Moreover, this paper derives the representation formula \eqref{eq22} for the unit outer co-normal vector to the boundary of the surface $\Gamma (t)$.

\begin{theorem}[Applications of surface divergence theorem]\label{thm24}{ \ }\\
Let $\Gamma (t)$ be an evolving surface. Let $\Gamma_\mathcal{A} (t)$, $\Gamma_\mathcal{B} (t)$, $\Gamma_\mathcal{C} (t)$ be three evolving surfaces with a piecewise Lipschitz-continuous boundary. The symbols $\nu_\mathcal{A}$, $\nu_\mathcal{B}$, and $\nu_\mathcal{C}$ denote the unit outer co-normal vector to $\partial \Gamma_\mathcal{A}$, $\partial \Gamma_\mathcal{B}$, and $\partial \Gamma_\mathcal{C}$, respectively. Suppose that $\nu_{\mathcal{A}} = - \nu_{\mathcal{B}} \text{ on } \partial \Gamma_{\mathcal{A}} (t) \cap \partial \Gamma_{\mathcal{B}} (t)$,
\begin{align*}
\nu_{\mathcal{B}}= - \nu_{\mathcal{C}} \text{ on }\partial \Gamma_{\mathcal{B}} (t) \cap \partial \Gamma_{\mathcal{C}} (t),{ \ }\nu_{\mathcal{C}} = - \nu_{\mathcal{A}}  \text{ on }\partial \Gamma_{\mathcal{C}} (t) \cap \partial \Gamma_{\mathcal{A}} (t),\\
\Gamma_{\mathcal{A}} (t) \cap \Gamma_\mathcal{B} (t) = \Gamma_{\mathcal{B}} (t) \cap \Gamma_{\mathcal{C}} (t) = \Gamma_{\mathcal{C}} (t) \cap \Gamma_{\mathcal{A}} (t) = \emptyset.
\end{align*}
If
\begin{equation*}
\Gamma (t) = \overset{\circ}{\overline{ \Gamma_{\mathcal{A}} (t) \cup \Gamma_{\mathcal{B}} (t)} } \text{ or }\Gamma (t) = \overset{\circ}{\overline{ \Gamma_{\mathcal{A}} (t) \cup \Gamma_{\mathcal{B}} (t) \cup \Gamma_{\mathcal{C}} (t)  } }, 
\end{equation*}
then \eqref{eq21} holds for all $\varphi = { }^t ( \varphi_1 , \varphi_2 , \varphi_3 ) \in C^1 ( \overline{ \Gamma (t)})$. Here $\nu := { }^t (0,0,0)$ if $\partial \Gamma (t) = \emptyset$, $\overline{\mathcal{M}} $ is the closure of the set $\mathcal{M}$, and $\overset{\circ}{\mathcal{M}}$ is the open kernel of the set $\mathcal{M}$. 
\end{theorem}

Applying Theorem \ref{thm24}, we can derive the divergence theorem on evolving surfaces with boundaries and evolving closed surfaces, and can calculate a variation of our dissipation energies.
\begin{theorem}[Variation of dissipation energy]\label{thm25}{ \ }\\
Let $e_{D}$ be a $C^1$-function. Fix $0 \leq t <T$ and $f \in C^2 ( \overline{\Gamma (t) } )$. For every $- 1 < \varepsilon < 1$ and $\psi \in C_0^1 ( \Gamma (t))$,
\begin{equation*}
E_{D}[f + \varepsilon \psi ](t) := - \int_{\Gamma (t)} \frac{1}{2} e_{D} ( | \nabla_\Gamma (f + \varepsilon \psi ) |^2) { \ } d \mathcal{H}^2_x .
\end{equation*}
Then
\begin{equation*}
\frac{d}{d \varepsilon} \bigg|_{ \varepsilon = 0} E_{D} [ f + \varepsilon \psi ] = \int_{\Gamma (t)} {\rm{div}}_\Gamma \{ e_{D}'(|{\rm{grad}}_\Gamma f |^2 ) {\rm{grad}}_\Gamma f \} \psi { \ } d \mathcal{H}^2_x .
\end{equation*}
\end{theorem}

Next we state the surface transport theorem. 
\begin{proposition}[Surface transport theorem]\label{prop26}{ \ }\\
Let $f \in C^1 ( \Gamma_T)$. For each $\Omega (t) \subset \Gamma (t)$, 
\begin{equation*}
\frac{d }{d t} \int_{\Omega (t) } f ( x , t ) { \ } d \mathcal{H}^2_x = \int_{\Omega (t)} \{ D_t^w \rho + ({\rm{div}}_\Gamma w) \rho \} (x,t) { \ }d \mathcal{H}^2_x . 
\end{equation*}
Here $D_t^w = \partial_t f + (w \cdot \nabla ) f $.
\end{proposition}
\noindent The proof of Proposition \ref{prop26} can be founded in Betounes \cite{Bet86}, Gurtin-Struthers-Williams \cite{GSW89}, Dziuk-Elliott \cite{DE07}, Koba-Liu-Giga \cite{KLG17}, and Koba (\cite{K18}, \cite{K19}). In Section \ref{sect5} we study the transport theorem on an evolving double bubble.

Applying an energetic variational approach with Theorem \ref{thm25} and Proposition \ref{prop26}, we derive the two systems \eqref{eq11} and \eqref{eq12}. See Section \ref{sect4} for details. Finally, we state the boundary conditions for the two systems.
\begin{theorem}[Conservation and energy laws of \eqref{eq11} and \eqref{eq12}]\label{thm27}{ \ }\\
$(\mathrm{i})$ Assume that 
\begin{equation*}
\frac{\partial C}{\partial \nu} \bigg|_{\partial \Gamma (t)} = 0.
\end{equation*}
Then the system \eqref{eq11} satisfies \eqref{eq13} and \eqref{eq14}.\\
$(\mathrm{ii})$ Assume that $C|_{\partial \Gamma (t)} = 0$. Then the system \eqref{eq11} satisfies \eqref{eq14}.\\
$(\mathrm{iii})$ Assume that 
\begin{equation*}
\frac{\partial \theta}{\partial \nu} \bigg|_{\partial \Gamma (t)} = 0.
\end{equation*}
Then the system \eqref{eq12} satisfies that for $t_1 < t_2$,
\begin{equation*}
\int_{\Gamma (t_2)} \rho (x,t_2) \theta (x , t_2) { \ } d \mathcal{H}^2_x = \int_{\Gamma (t_1)} \rho (x,t_1) \theta (x, t_1) { \ }d \mathcal{H}^2_x,
\end{equation*}
and
\begin{multline}\label{eq23}
 \int_{\Gamma (t_2)} \frac{1}{2} \rho (x,t_2) | \theta (x , t_2)|^2 { \ } d \mathcal{H}^2_x + \int_{t_1}^{t_2} \int_{\Gamma (\tau )} e_2' (| \nabla_\Gamma \theta |^2 ) | \nabla_\Gamma \theta |^2 { \ }d \mathcal{H}^2_x d \tau\\
 = \int_{\Gamma (t_1)} \frac{1}{2} \rho (x,t_1) | \theta (x, t_1)|^2 { \ }d \mathcal{H}^2_x.
\end{multline}
$(\mathrm{iv})$ Assume that $\theta |_{\partial \Gamma (t)} = 0$. Then the system \eqref{eq12} satisfies \eqref{eq23}.
\end{theorem}

In Section \ref{sect3} we prove Theorems \ref{thm23} and \ref{thm24}. In Section \ref{sect4} we derive the two system \eqref{eq11} and \eqref{eq12} from an energetic point of view. In subsection \ref{subsec41} we give the proof of Theorem \ref{thm25}. In subsection \ref{subsec42} we show the assertions $(\mathrm{i})$ and $(\mathrm{ii})$ of Theorem \ref{thm27}. In subsection \ref{subsec43} we prove the assertions $(\mathrm{iii})$ and $(\mathrm{iv})$ of Theorem \ref{thm27}. In Section \ref{sect5} we study the divergence and transport theorems on an evolving double bubble. In Section \ref{sect6} we make a mathematical model for a diffusion process on an evolving double bubble, and study the conservation and energy laws of the system \eqref{eq17}.

\section{Surface Divergence Theorem}\label{sect3}

This section is devoted to deriving the surface divergence theorem and the representation formula for the unit outer co-normal vector to the boundary of the surface $\Gamma (t)$. To prove Theorems \ref{thm23} and \ref{thm24}, we prepare the following three lemmas.

\begin{lemma}[Representation formulas for surface divergence]\label{lem31}{ \ }\\
$(\mathrm{i})$ For every $\varphi = { }^t ( \varphi_1 , \varphi_2 , \varphi_3 ) \in C^1 ( \overline{ \Gamma (t)})$,
\begin{equation}\label{eq31}
\int_{\Gamma (t)} {\rm{div}}_\Gamma \varphi { \ }d \mathcal{H}^2_x = \int_{U} \mathfrak{g}^{\alpha} \cdot \frac{\partial \varphi}{\partial X_\alpha}  \sqrt{ \mathcal{G} } { \ } d X.
\end{equation}
$(\mathrm{ii})$ For each $0 < t < T$,
\begin{equation}\label{eq32}
\int_{\Gamma (t)} {\rm{div}}_\Gamma w{ \ }d \mathcal{H}^2_x = \int_{U} \frac{\partial }{\partial t }\sqrt{ \mathcal{G} } { \ } d X.
\end{equation}
Here $\mathcal{G} = | \mathfrak{g}_1 \times \mathfrak{g}_2|^2 = \mathfrak{g}_{11} \mathfrak{g}_{22} - \mathfrak{g}_{12} \mathfrak{g}_{21}$.
\end{lemma}

\begin{lemma}\label{lem32}
For every $0 \leq t < T$ and $X \in \overline{U}$, we set
\begin{equation*}
\mathfrak{g}_3 = \mathfrak{g}_3 ( X , t) : = \frac{ \mathfrak{g}_1 \times \mathfrak{g}_2 }{ | \mathfrak{g}_1 \times \mathfrak{g}_2| }.
\end{equation*}
Then the following three assertions hold:\\
$(\mathrm{i})$ For every $0 \leq t< T$ and $X \in \overline{U}$, 
\begin{align}
\mathfrak{g}^1 \sqrt{\mathfrak{g}_{11} \mathfrak{g}_{22} - \mathfrak{g}_{12} \mathfrak{g}_{21} } & = \mathfrak{g}_2 \times \mathfrak{g}_3,\label{eq33}\\
\mathfrak{g}^2 \sqrt{\mathfrak{g}_{11} \mathfrak{g}_{22} - \mathfrak{g}_{12} \mathfrak{g}_{21} } & = - \mathfrak{g}_1 \times \mathfrak{g}_3.\label{eq34}
\end{align}
$(\mathrm{ii})$ $($\text{Weingarten formula}$)$: For all $0 \leq t< T$ and $X \in \overline{U}$,
\begin{align}
\frac{\partial \mathfrak{g}_3}{\partial X_1} = c_1 \mathfrak{g}_1 + c_2 \mathfrak{g}_2,\label{eq35}\\
\frac{\partial \mathfrak{g}_3}{\partial X_2} = c_3 \mathfrak{g}_1 + c_4 \mathfrak{g}_2,\label{eq36}
\end{align}
where
\begin{align*}
c_1 := \frac{ \mathfrak{g}_{12} \left( \frac{\partial \mathfrak{g}_1}{\partial X_2} \cdot \mathfrak{g}_3 \right) - \mathfrak{g}_{22} \left( \frac{\partial \mathfrak{g}_1}{\partial X_1} \cdot \mathfrak{g}_3 \right) }{ \mathfrak{g}_{11} \mathfrak{g}_{22} - \mathfrak{g}_{12} \mathfrak{g}_{21} }, c_2 := \frac{ \mathfrak{g}_{12} \left( \frac{\partial \mathfrak{g}_1}{\partial X_1} \cdot \mathfrak{g}_3 \right) - \mathfrak{g}_{11} \left( \frac{\partial \mathfrak{g}_1}{\partial X_2} \cdot \mathfrak{g}_3 \right) }{ \mathfrak{g}_{11} \mathfrak{g}_{22} - \mathfrak{g}_{12} \mathfrak{g}_{21} },\\
c_3 := \frac{ \mathfrak{g}_{12} \left( \frac{\partial \mathfrak{g}_2}{\partial X_2} \cdot \mathfrak{g}_3 \right) - \mathfrak{g}_{22} \left( \frac{\partial \mathfrak{g}_1}{\partial X_2} \cdot \mathfrak{g}_3 \right) }{ \mathfrak{g}_{11} \mathfrak{g}_{22} - \mathfrak{g}_{12} \mathfrak{g}_{21} }, 
c_4 := \frac{ \mathfrak{g}_{12} \left( \frac{\partial \mathfrak{g}_1}{\partial X_2} \cdot \mathfrak{g}_3 \right) - \mathfrak{g}_{11} \left( \frac{\partial \mathfrak{g}_2}{\partial X_2} \cdot \mathfrak{g}_3 \right) }{ \mathfrak{g}_{11} \mathfrak{g}_{22} - \mathfrak{g}_{12} \mathfrak{g}_{21} }.
\end{align*}
$(\mathrm{iii})$ 
Fix $0 \leq t < T$. Then for every $\varphi = { }^t ( \varphi_1 , \varphi_2 , \varphi_3 ) \in [ C^1 ( \overline{\Gamma (t)} ) ]^3$,
\begin{equation}\label{eq37}
- \int_{U} \left\{   \mathfrak{g}_2 \times  \frac{\partial \mathfrak{g}_3}{\partial X_1} - \mathfrak{g}_1 \times  \frac{\partial \mathfrak{g}_3}{\partial X_2} \right\} \cdot \varphi { \ } d X = - \int_{\Gamma (t)} H_\Gamma (n \cdot \varphi ) { \ }d \mathcal{H}^2_x.
\end{equation}
\end{lemma}

\begin{lemma}\label{lem33}
Fix $0 \leq t < T$. Then for every $f \in C( \overline{\Gamma (t)})$,
\begin{equation}\label{eq38}
\int_{\partial \Gamma (t)} f { \ }d \mathcal{H}^1_x = \int_{\partial U} f | n_1^U \mathfrak{g}_2 - n_2^U \mathfrak{g}_1 | { \ }d \mathcal{H}^1_X.
\end{equation}
Here $n^U = (n_1^U , n_2^U )$ is the unit outer normal vector to $\partial U$.
\end{lemma}
\noindent Since the proof of Lemma \ref{lem31} can be founded in \cite{K18} and \cite{K19}, we now prove Lemmas \ref{lem32} and \ref{lem33}.
\begin{proof}[Proof of Lemma \ref{lem32}]
We first show $(\mathrm{i})$. It is easy to check that
\begin{multline*}
\mathfrak{g}^1 \sqrt{ \mathfrak{g}_{11} \mathfrak{g}_{22} - \mathfrak{g}_{12} \mathfrak{g}_{21}}  = ( \mathfrak{g}^{11} \mathfrak{g}_1 + \mathfrak{g}^{12} \mathfrak{g}_2 )\sqrt{ \mathfrak{g}_{11} \mathfrak{g}_{22} - \mathfrak{g}_{12} \mathfrak{g}_{21}}  \\
= \frac{1}{ | \mathfrak{g}_1 \times \mathfrak{g}_2 | }
\begin{pmatrix}
\frac{\partial \widehat{x}_2}{\partial X_2}  \left( \frac{\partial \widehat{x}_1}{\partial X_1}  \frac{\partial \widehat{x}_2}{\partial X_2}  - \frac{\partial \widehat{x}_1}{\partial X_2}  \frac{\partial \widehat{x}_2}{\partial X_1}  \right) -\frac{\partial \widehat{x}_3}{\partial X_2}  \left( \frac{\partial \widehat{x}_1}{\partial X_2}  \frac{\partial \widehat{x}_3}{\partial X_1}  - \frac{\partial \widehat{x}_1}{\partial X_1}  \frac{\partial \widehat{x}_3}{\partial X_2}  \right)\\
\frac{\partial \widehat{x}_2}{\partial X_2}  \left( \frac{\partial \widehat{x}_2}{\partial X_1}  \frac{\partial \widehat{x}_3}{\partial X_2}  - \frac{\partial \widehat{x}_2}{\partial X_2} \frac{\partial \widehat{x}_3}{\partial X_1}  \right) - \frac{\partial \widehat{x}_1}{\partial X_2}  \left( \frac{\partial \widehat{x}_1}{\partial X_2}  \frac{\partial \widehat{x}_2}{\partial X_1}  - \frac{\partial \widehat{x}_1}{\partial X_1}  \frac{\partial \widehat{x}_2}{\partial X_2}  \right)\\
\frac{\partial \widehat{x}_1}{\partial X_2}  \left( \frac{\partial \widehat{x}_1}{\partial X_2}  \frac{\partial \widehat{x}_3}{\partial X_1}  - \frac{\partial \widehat{x}_1}{\partial X_1}  \frac{\partial \widehat{x}_3}{\partial X_2}  \right) - \frac{\partial \widehat{x}_2}{\partial X_2}  \left( \frac{\partial \widehat{x}_2}{\partial X_2}  \frac{\partial \widehat{x}_3}{\partial X_1}  - \frac{\partial \widehat{x}_2}{\partial X_1}  \frac{\partial \widehat{x}_3}{\partial X_2}  \right)
\end{pmatrix}\\
= \mathfrak{g}_2 \times \frac{ \mathfrak{g}_1 \times \mathfrak{g}_2 }{ | \mathfrak{g}_1 \times \mathfrak{g}_2| }
\end{multline*}
and that
\begin{multline*}
\mathfrak{g}^2 \sqrt{ \mathfrak{g}_{11} \mathfrak{g}_{22} - \mathfrak{g}_{12} \mathfrak{g}_{21}}  = ( \mathfrak{g}^{21} \mathfrak{g}_1 + \mathfrak{g}^{22} \mathfrak{g}_2 )\sqrt{ \mathfrak{g}_{11} \mathfrak{g}_{22} - \mathfrak{g}_{12} \mathfrak{g}_{21}}  \\
= \frac{1}{ | \mathfrak{g}_1 \times \mathfrak{g}_2 | }
\begin{pmatrix}
\frac{\partial \widehat{x}_2}{\partial X_1}  \left( \frac{\partial \widehat{x}_1}{\partial X_1}  \frac{\partial \widehat{x}_2}{\partial X_2}  - \frac{\partial \widehat{x}_1}{\partial X_2}  \frac{\partial \widehat{x}_2}{\partial X_1}  \right) - \frac{\partial \widehat{x}_3}{\partial X_1}  \left( \frac{\partial \widehat{x}_1}{\partial X_2}  \frac{\partial \widehat{x}_3}{\partial X_1}  - \frac{\partial \widehat{x}_1}{\partial X_1}  \frac{\partial \widehat{x}_3}{\partial X_2}  \right)\\
\frac{\partial \widehat{x}_3}{\partial X_1}  \left( \frac{\partial \widehat{x}_2}{\partial X_1}  \frac{\partial \widehat{x}_3}{\partial X_2}  - \frac{\partial \widehat{x}_2}{\partial X_2}  \frac{\partial \widehat{x}_3}{\partial X_1}  \right) - \frac{\partial \widehat{x}_1}{\partial X_1}  \left(\frac{\partial \widehat{x}_1}{\partial X_2}  \frac{\partial \widehat{x}_2}{\partial X_1}  - \frac{\partial \widehat{x}_1}{\partial X_1}  \frac{\partial \widehat{x}_2}{\partial X_2}  \right)\\
\frac{\partial \widehat{x}_1}{\partial X_1}  \left( \frac{\partial \widehat{x}_1}{\partial X_2}  \frac{\partial \widehat{x}_3}{\partial X_1}  - \frac{\partial \widehat{x}_1}{\partial X_1}  \frac{\partial \widehat{x}_3}{\partial X_2}  \right) - \frac{\partial \widehat{x}_2}{\partial X_1}  \left( \frac{\partial \widehat{x}_2}{\partial X_2}  \frac{\partial \widehat{x}_3}{\partial X_1}  - \frac{\partial \widehat{x}_2}{\partial X_1}  \frac{\partial \widehat{x}_3}{\partial X_2}  \right)
\end{pmatrix} \\
= - \mathfrak{g}_1 \times \frac{ \mathfrak{g}_1 \times \mathfrak{g}_2 }{ | \mathfrak{g}_1 \times \mathfrak{g}_2| }.
\end{multline*}
Therefore we have \eqref{eq33} and \eqref{eq34}.

Next we prove $(\mathrm{ii})$. Fix $0 \leq t< T$ and $X \in \overline{U}$. Since $\mathfrak{g}_1$, $\mathfrak{g}_2$, $\mathfrak{g}_3$ is a basis for $\mathbb{R}^3$, there are $c_1, c_2,c_3,c_4, c_5, c_6 \in \mathbb{R}$ such that
\begin{align}
\frac{\partial \mathfrak{g}_3}{\partial X_1} = c_1 \mathfrak{g}_1 + c_2 \mathfrak{g}_2 + c_5 \mathfrak{g}_3,\label{eq39}\\
\frac{\partial \mathfrak{g}_3}{\partial X_2} = c_3 \mathfrak{g}_1 + c_4 \mathfrak{g}_2 + c_6 \mathfrak{g}_6.\label{Eq310}
\end{align}
From 
\begin{equation}\label{Eq311}
\mathfrak{g}_1 \cdot \mathfrak{g}_3 = 0, { \ \ \ }\mathfrak{g}_2 \cdot \mathfrak{g}_3 = 0, \text{  and }\mathfrak{g}_3 \cdot \mathfrak{g}_3 = 1,
\end{equation}
we find that
\begin{multline}\label{Eq312}
\mathfrak{g}_1 \cdot \frac{\partial \mathfrak{g}_3}{\partial X_1} = - \frac{\partial \mathfrak{g}_1}{\partial X_1} \cdot \mathfrak{g}_3,{ \ \ \ }\mathfrak{g}_1 \cdot \frac{\partial \mathfrak{g}_3}{\partial X_2} = - \frac{\partial \mathfrak{g}_1}{\partial X_2} \cdot \mathfrak{g}_3,{ \ \ \ }\mathfrak{g}_2 \cdot \frac{\partial \mathfrak{g}_3}{\partial X_1} = - \frac{\partial \mathfrak{g}_2}{\partial X_1} \cdot \mathfrak{g}_3,\\
\mathfrak{g}_2 \cdot \frac{\partial \mathfrak{g}_3}{\partial X_2} = - \frac{\partial \mathfrak{g}_2}{\partial X_2} \cdot \mathfrak{g}_3, { \ \ \ }\mathfrak{g}_3 \cdot \frac{\partial \mathfrak{g}_3}{\partial X_1} = 0, { \ \ \ } \mathfrak{g}_3 \cdot \frac{\partial \mathfrak{g}_3}{\partial X_2} = 0.
\end{multline}
Using \eqref{eq39}, \eqref{Eq311}, and \eqref{Eq312}, we observe that
\begin{align*}
0 = \frac{\partial \mathfrak{g}_3}{\partial X_1} \cdot \mathfrak{g}_3 = c_1 (\mathfrak{g}_1 \cdot \mathfrak{g}_3) + c_2 (\mathfrak{g}_2 \cdot \mathfrak{g}_3) + c_5 ( \mathfrak{g}_3 \cdot \mathfrak{g}_3) = c_5,\\
- \frac{\partial \mathfrak{g}_1}{\partial X_1} \cdot \mathfrak{g}_3 = \frac{\partial \mathfrak{g}_3}{\partial X_1} \cdot \mathfrak{g}_1 = c_1 (\mathfrak{g}_1 \cdot \mathfrak{g}_1) + c_2 (\mathfrak{g}_2 \cdot \mathfrak{g}_1) = c_1 \mathfrak{g}_{11} + c_2 \mathfrak{g}_{12},\\
- \frac{\partial \mathfrak{g}_2}{\partial X_1} \cdot \mathfrak{g}_3 = \frac{\partial \mathfrak{g}_3}{\partial X_1} \cdot \mathfrak{g}_2 = c_1 (\mathfrak{g}_1 \cdot \mathfrak{g}_2) + c_2 (\mathfrak{g}_2 \cdot \mathfrak{g}_2) = c_1 \mathfrak{g}_{12} + c_2 \mathfrak{g}_{22}.
\end{align*}
This implies that $c_5 = 0$,
\begin{align*}
c_1 = \frac{ \mathfrak{g}_{12} \left( \frac{\partial \mathfrak{g}_1}{\partial X_2} \cdot \mathfrak{g}_3 \right) - \mathfrak{g}_{22} \left( \frac{\partial \mathfrak{g}_1}{\partial X_1} \cdot \mathfrak{g}_3 \right) }{ \mathfrak{g}_{11} \mathfrak{g}_{22} - \mathfrak{g}_{12} \mathfrak{g}_{21} }, { \ }c_2 = \frac{ \mathfrak{g}_{12} \left( \frac{\partial \mathfrak{g}_1}{\partial X_1} \cdot \mathfrak{g}_3 \right) - \mathfrak{g}_{11} \left( \frac{\partial \mathfrak{g}_1}{\partial X_2} \cdot \mathfrak{g}_3 \right) }{ \mathfrak{g}_{11} \mathfrak{g}_{22} - \mathfrak{g}_{12} \mathfrak{g}_{21} }.
\end{align*}
Since $c_1$, $c_2$, $c_5$ are continuous functions with respect to $(X,t)$, \eqref{eq35} holds for all $0 \leq t< T$ and $X \in \overline{U}$. Similarly, we obtain \eqref{eq36}. Therefore the assertion $(\mathrm{ii})$ is proved.

Finally, we derive \eqref{eq37}. Fix $\varphi = { }^t ( \varphi_1 , \varphi_2 , \varphi_3 ) \in [ C^1 ( \overline{\Gamma (t)} ) ]^3$. We first consider (L.H.S) of \eqref{eq37}. Using \eqref{eq35} and \eqref{eq36}, we see that
\begin{align*}
- \left( \mathfrak{g}_2 \times  \frac{\partial \mathfrak{g}_3}{\partial X_1} - \mathfrak{g}_1 \times  \frac{\partial \mathfrak{g}_3}{\partial X_2} \right) =  - c_1 (\mathfrak{g}_2 \times \mathfrak{g}_1) + c_4 (\mathfrak{g}_1 \times \mathfrak{g}_2)\\
= | \mathfrak{g}_1 \times \mathfrak{g}_2 | ( c_1 + c_4) \frac{ \mathfrak{g}_1 \times \mathfrak{g}_2}{| \mathfrak{g}_1 \times \mathfrak{g}_2 | }. 
\end{align*}
Therefore we have
\begin{equation}\label{Eq313}
(\text{L.H.S.}) \text{ of }\eqref{eq37}   = \int_{U} | \mathfrak{g}_1 \times \mathfrak{g}_2 | ( c_1 + c_4) \left( \mathfrak{g}_3 \cdot \varphi \right) { \ } d X.
\end{equation}
Next we consider (R.H.S.) of \eqref{eq37}. From \eqref{eq31}, \eqref{eq33}, and \eqref{eq34}, we check that
\begin{multline*}
- \int_{\Gamma (t)} H_\Gamma ( n \cdot \varphi ) { \ }d \mathcal{H}^2_x = - \int_{\Gamma (t)}- ({\rm{div}}_\Gamma n ) ( n \cdot \varphi ) { \ }d \mathcal{H}^2_x\\
= \int_U \left( \frac{ \mathfrak{g}_1 \times \mathfrak{g}_2 }{ | \mathfrak{g}_1 \times \mathfrak{g}_2| } \cdot \varphi \right) \mathfrak{g}^\alpha \cdot \left( \frac{\partial }{\partial X_\alpha} \frac{ \mathfrak{g}_1 \times \mathfrak{g}_2 }{ | \mathfrak{g}_1 \times \mathfrak{g}_2| } \right) \sqrt{ \mathfrak{g}_{11} \mathfrak{g}_{22} - \mathfrak{g}_{12} \mathfrak{g}_{21} } { \ }d X\\
= \int_U \left( \mathfrak{g}_3 \cdot \varphi \right) \left\{ (\mathfrak{g}_2 \times \mathfrak{g}_3) \cdot \left( \frac{\partial \mathfrak{g}_3}{\partial X_1} \right)  - (\mathfrak{g}_1 \times \mathfrak{g}_3)  \cdot \left( \frac{\partial \mathfrak{g}_3 }{\partial X_2}\right) \right\} { \ }d X.
\end{multline*}
Using \eqref{eq35} and \eqref{eq36}, we find that
\begin{multline*}
 (\mathfrak{g}_2 \times \mathfrak{g}_3) \cdot \left( \frac{\partial \mathfrak{g}_3}{\partial X_1} \right)  - (\mathfrak{g}_1 \times \mathfrak{g}_3)  \cdot \left( \frac{\partial \mathfrak{g}_3 }{\partial X_2}\right) =  c_1 (\mathfrak{g}_2 \times \mathfrak{g}_3) \cdot \mathfrak{g}_1  + c_4 (\mathfrak{g}_1 \times \mathfrak{g}_3)  \cdot \mathfrak{g}_2\\
 = \frac{c_1}{|\mathfrak{g}_1 \times \mathfrak{g}_2 |} \{ \mathfrak{g}_2 \times (\mathfrak{g}_1 \times \mathfrak{g}_2) \} \cdot \mathfrak{g}_1 + \frac{c_4}{|\mathfrak{g}_1 \times \mathfrak{g}_2 |} \{ \mathfrak{g}_2 \times (\mathfrak{g}_1 \times \mathfrak{g}_2) \} \cdot \mathfrak{g}_2\\
=  \frac{c_1 + c_4}{|\mathfrak{g}_1 \times \mathfrak{g}_2 |} \{ ( \mathfrak{g}_1 \cdot \mathfrak{g}_1) ( \mathfrak{g}_2 \cdot \mathfrak{g}_2) - (\mathfrak{g}_1 \cdot \mathfrak{g}_2)^2 \} = (c_1 + c_4) | \mathfrak{g}_1 \times \mathfrak{g}_2|.
\end{multline*}
Here we used the facts that $v_1 \times ( v_2 \times v_3) = (v_1 \cdot v_3) v_2 - (v_1 \cdot v_2) v_3$ $(v_1,v_2,v_3 \in \mathbb{R}^3)$ and that $| \mathfrak{g}_1 \times \mathfrak{g}_2|^2 = \mathfrak{g}_{11} \mathfrak{g}_{22} - \mathfrak{g}_{12} \mathfrak{g}_{21} $. Therefore we have
\begin{equation}\label{Eq314}
(\text{R.H.S.}) \text{ of }\eqref{eq37}   = \int_{U} | \mathfrak{g}_1 \times \mathfrak{g}_2 | ( c_1 + c_4) \left( \mathfrak{g}_3 \cdot \varphi \right) { \ } d X.
\end{equation}
From \eqref{Eq313} and \eqref{Eq314} we have \eqref{eq37}. Therefore the lemma follows. 
\end{proof}

\begin{proof}[Proof of Lemma \ref{lem33}]
Fix $f , \psi \in C (\overline{\Gamma (t)} )$. We first consider (R.H.S.) of \eqref{eq38}. By the formula for the line integral and Definition \ref{def21}, we see that
\begin{equation}
\int_{\partial U} \psi { \ } d \mathcal{H}_X^1 = \sum_{m=1}^M \int_{\partial U^m} \psi { \ } d \mathcal{H}_X^1 = \sum_{m=1}^M \int_{a_m}^{b_m} \psi \sqrt{ \left( p_m' \right)^2 + \left( q_m' \right)^2 } { \ }d \ell \label{Eq315} .
\end{equation}
Here $p_m' = p'_m (\ell ) = \frac{ d p_m}{d \ell} (\ell)$ and $q_m' = q_m' ( \ell ) = \frac{ d q_m}{d \ell}( \ell )$.

Now we study the unit outer normal vector $n^U$ to $\partial U$. Since
\begin{equation*}
\partial U^m = \{ X \in \mathbb{R}^2; { \ }X_1 = p_m (\ell ) , X_2 = q_m ( \ell ) , \ell \in [a_m , b_m ] \},
\end{equation*}
we check that for almost all $\ell \in [a_m , b_m]$,
\begin{equation}\label{Eq316}
n^U( p_m (\ell ) ,  q_m (\ell )) = \left( \mp \frac{q_m' }{ \sqrt{ \left( p_m' \right)^2 + \left( q_m' \right)^2 }} ,  \pm \frac{p_m' }{ \sqrt{ \left( p_m' \right)^2 + \left( q_m' \right)^2 }}  \right).
\end{equation}
Next we consider (L.H.S.) of \eqref{eq38}. Since
\begin{align*}
\partial \Gamma^m (t) &= \{ x \in \mathbb{R}^3; { \ }x = \widehat{x} (X, t) , X \in \partial U^m \}\\
& = \{ x \in \mathbb{R}^3; { \ }x = \widehat{x} ( p_m (\ell ) , q_m (\ell ) , t) , \ell \in [a_m , b_m] \},
\end{align*}
we check that
\begin{multline*}
\int_{\partial \Gamma (t)} f { \ }d \mathcal{H}^1_x = \sum_{m=1}^M \int_{a_m}^{b_m} f \sqrt{ \left( \frac{d \widehat{x}_1}{d \ell }\right)^2 + \left( \frac{d \widehat{x}_3}{d \ell }\right)^2 + \left( \frac{d \widehat{x}_3}{d \ell }\right)^2   } { \ }d \ell\\
= \sum_{m=1}^M \int_{a_m}^{b_m} f \sqrt{ \left( p_m' \right)^2 \mathfrak{g}_{11} + 2 \left( p_m' \right) \left( q_m' \right) \mathfrak{g}_{12} +  \left( q_m' \right)^2 \mathfrak{g}_{22}   } { \ }d \ell .
\end{multline*}
Using \eqref{Eq316}, we obtain
\begin{equation}\label{Eq317}
\int_{\partial \Gamma (t)} f { \ }d \mathcal{H}^1_x = \sum_{m=1}^M \int_{a_m}^{b_m} f  | n_1^U \mathfrak{g}_2 - n_2^U \mathfrak{g}_1 | \sqrt{ \left( p_m' \right)^2 + \left( q_m' \right)^2 } { \ }d \ell .
\end{equation}
From \eqref{Eq315} and \eqref{Eq317}, we have \eqref{eq38}. Therefore Lemma \ref{lem33} is proved. 
\end{proof}

Let us attack Theorems \ref{thm23} and \ref{thm24}. 
\begin{proof}[Proof of Theorem \ref{thm23}]
Fix $ 0 \leq t < T$ and $\varphi = { }^t (\varphi_1 , \varphi_2 , \varphi_3 ) \in C^1 ( \overline{\Gamma (t)} )$. Using \eqref{eq31}, \eqref{eq33}, and \eqref{eq34}, we see that
\begin{multline*}
\int_{\Gamma (t)} {\rm{div}}_\Gamma \varphi { \ }d \mathcal{H}^2_x = \int_{U} \mathfrak{g}^{\alpha} \cdot \frac{\partial \varphi}{\partial X_\alpha}  \sqrt{ \mathfrak{g}_{11} \mathfrak{g}_{22} - \mathfrak{g}_{12} \mathfrak{g}_{21}} { \ } d X\\
= \int_{U} \left( \mathfrak{g}^1 \cdot \frac{\partial \varphi}{\partial X_1} + \mathfrak{g}^2 \cdot \frac{\partial \varphi}{\partial X_2} \right) \sqrt{ \mathfrak{g}_{11} \mathfrak{g}_{22} - \mathfrak{g}_{12} \mathfrak{g}_{21}} { \ } d X\\
= \int_{U} \left\{ \left( \mathfrak{g}_2 \times \frac{ \mathfrak{g}_1 \times \mathfrak{g}_2}{ |\mathfrak{g}_1 \times \mathfrak{g}_2| } \right) \cdot \frac{\partial \varphi}{\partial X_1} - \left( \mathfrak{g}_1 \times  \frac{ \mathfrak{g}_1 \times \mathfrak{g}_2}{ |\mathfrak{g}_1 \times \mathfrak{g}_2| } \right) \cdot \frac{\partial \varphi}{\partial X_2}  \right\} { \ } d X.
\end{multline*}
Using the Gauss divergence theorem, we have
\begin{multline}\label{Eq318}
\int_{\Gamma (t)} {\rm{div}}_\Gamma \varphi { \ }d \mathcal{H}^2_x = \int_{\partial U} \left\{ (n_1^U \mathfrak{g}_2 -n_2^U \mathfrak{g}_1) \times  \frac{ \mathfrak{g}_1 \times \mathfrak{g}_2}{ |\mathfrak{g}_1 \times \mathfrak{g}_2| } \right\} \cdot \varphi  { \ } d \mathcal{H}^1_X\\
- \int_{U} \left[ \left\{ \frac{\partial}{\partial X_1} \left( \mathfrak{g}_2 \times  \frac{ \mathfrak{g}_1 \times \mathfrak{g}_2}{ |\mathfrak{g}_1 \times \mathfrak{g}_2| } \right) \right\} \cdot \varphi - \frac{\partial}{\partial X_2} \left\{ \left( \mathfrak{g}_1 \times  \frac{ \mathfrak{g}_1 \times \mathfrak{g}_2}{ |\mathfrak{g}_1 \times \mathfrak{g}_2| } \right) \right\} \cdot \varphi  \right] { \ } d X\\
= : J_1 + J_2.
\end{multline}
From Lemma \ref{lem33}, we check that
\begin{align}
J_1 &= \int_{\partial U} \left( \frac{n_1^U \mathfrak{g}_2 - n_2^U \mathfrak{g}_1}{ | n_1^U \mathfrak{g}_2 - n_2^U \mathfrak{g}_1  | } \times \frac{\mathfrak{g}_1 \times \mathfrak{g}_2}{| \mathfrak{g}_1 \times \mathfrak{g}_2|} \right) \cdot \varphi | n_1^U \mathfrak{g}_2 - n_2^U \mathfrak{g}_1  |{ \ } d \mathcal{H}^1_X \notag\\
&= \int_{\partial \Gamma (t)} \nu \cdot \varphi { \ }d \mathcal{H}^1_x. \label{Eq319}
\end{align}
Next we study $J_2$. Since $\partial \mathfrak{g}_1/{\partial X_2} = \partial \mathfrak{g}_2/ {\partial X_1}$, we use \eqref{eq37} to have
\begin{equation}\label{Eq320}
J_2 =- \int_{U} \left\{   \mathfrak{g}_2 \times  \frac{\partial \mathfrak{g}_3}{\partial X_1} - \mathfrak{g}_1 \times  \frac{\partial \mathfrak{g}_3}{\partial X_2} \right\} \cdot \varphi { \ } d X = - \int_{\Gamma (t)} H_\Gamma (n \cdot \varphi ) { \ }d \mathcal{H}^2_x.
\end{equation}
Combing \eqref{Eq318}, \eqref{Eq319}, and \eqref{Eq320}, we see that
\begin{equation*}
\int_{\Gamma (t)} {\rm{div}}_\Gamma \varphi { \ }d \mathcal{H}^2_x = - \int_{\Gamma (t)} H_\Gamma (n \cdot \varphi ) { \ } d \mathcal{H}^2_x + \int_{\partial \Gamma (t)} \nu \cdot \varphi { \ }d \mathcal{H}^1_x.
\end{equation*}
From \eqref{Eq319} we have \eqref{eq22}. Therefore Theorem \ref{thm23} is proved. 
\end{proof}

\begin{proof}[Proof of Theorem \ref{thm24}]
We only consider the case when $\Gamma = \overset{\circ}{\overline{ \Gamma_{\mathcal{A}} \cup \Gamma_{\mathcal{B}} \cup \Gamma_{\mathcal{C}} } }$. Fix $\varphi \in [C ^1 ( \overline{\Gamma (t)} )]^3$. Applying Theorem \ref{thm23}, we check that
\begin{multline*}
\int_{\Gamma (t)} {\rm{div}}_\Gamma \varphi { \ }d \mathcal{H}^2_x = \int_{\Gamma_{\mathcal{A}} } {\rm{div}}_{\Gamma_\mathcal{A}} \varphi { \ }d \mathcal{H}^2_x + \int_{\Gamma_{\mathcal{B}} } {\rm{div}}_{\Gamma_{\mathcal{B}} } \varphi { \ }d \mathcal{H}^2_x + \int_{\Gamma_{\mathcal{C}}} {\rm{div}}_{\Gamma_{\mathcal{C}}} \varphi { \ }d \mathcal{H}^2_x\\
= - \int_{\Gamma_{\mathcal{A}} (t)} H_{\Gamma_\mathcal{A}} (n_\mathcal{A} \cdot \varphi ) { \ }d \mathcal{H}^2_x - \int_{\Gamma_{\mathcal{B}} (t)} H_{\Gamma_\mathcal{B}} (n_\mathcal{B} \cdot \varphi ) { \ }d \mathcal{H}^2_x - \int_{\Gamma_{\mathcal{C}} (t)} H_{\Gamma_\mathcal{C}} (n_\mathcal{C} \cdot \varphi ) { \ }d \mathcal{H}^2_x\\
 + \int_{\partial \Gamma_{\mathcal{A}} (t)} \nu_{\mathcal{A}} \cdot \varphi { \ }d \mathcal{H}^1_x + \int_{\partial \Gamma_{\mathcal{B}} (t)} \nu_{\mathcal{B}} \cdot \varphi { \ }d \mathcal{H}^1_x + \int_{\partial \Gamma_{\mathcal{C}} (t)} \nu_{\mathcal{C}} \cdot \varphi { \ }d \mathcal{H}^1_x.
\end{multline*}
Here $H_{\Gamma_\sharp} = - {\rm{div}}_{\Gamma_\sharp} n_\sharp$, $n_\sharp = { }^t (n_1^\sharp , n_2^\sharp , n_3^\sharp )$ is the unit outer normal vector to $\overline{\Gamma_\sharp}$, $\nu_\sharp = { }^t (\nu_1^\sharp , \nu_2^\sharp , \nu_3^\sharp )$ is the unit outer co-normal vector to $\partial \Gamma_\sharp$, where $\sharp = \mathcal{A}, \mathcal{B}, \mathcal{C}$.

From $\Gamma = \overset{\circ}{\overline{ \Gamma_{\mathcal{A}} \cup \Gamma_{\mathcal{B}} \cup \Gamma_{\mathcal{C}} } }$, $\Gamma_{\mathcal{A}} \cap \Gamma_{\mathcal{B}} = \Gamma_{\mathcal{B}} \cap \Gamma_{\mathcal{C}} = \Gamma_{\mathcal{C}} \cap \Gamma_{\mathcal{A}} = \emptyset$, we see that
\begin{equation*}
- \sum_{\sharp = \mathcal{A}, \mathcal{B}, \mathcal{C}} \int_{\Gamma_{\sharp} } H_{\Gamma_\sharp} (n_\sharp \cdot \varphi ) { \ }d \mathcal{H}^2_x = - \int_{\Gamma (t)} H_\Gamma (n \cdot \varphi ) { \ }d \mathcal{H}^2_x.
\end{equation*}
Since
\begin{multline*}
\int_{\partial \Gamma_{\flat} (t)} \nu_\flat \cdot \varphi { \ }d \mathcal{H}^1_x = \int_{\partial \Gamma_{\flat} \setminus ( \partial \Gamma_{\natural} \cup \partial \Gamma_{\sharp} )} \nu_{\flat} \cdot \varphi { \ }d \mathcal{H}^1_x \\
+ \int_{\partial \Gamma_{\flat} \cap \partial \Gamma_{\natural }} \nu_{\flat} \cdot \varphi { \ }d \mathcal{H}^1_x + \int_{\partial \Gamma_{\flat} \cap \partial \Gamma_{\sharp}} \nu_{\flat} \cdot \varphi { \ }d \mathcal{H}^1_x
\end{multline*}
for  $(\flat , \natural , \sharp ) = ( \mathcal{A} , \mathcal{B} , \mathcal{C})$, $(\flat , \natural , \sharp ) = ( \mathcal{B} , \mathcal{C} , \mathcal{A})$, $(\flat , \natural , \sharp ) = ( \mathcal{C} , \mathcal{A} , \mathcal{B})$, it follows from the assumptions that $\nu_{\mathcal{A}} = - \nu_{\mathcal{B}}$ on $\partial \Gamma_{\mathcal{A}} \cap \partial \Gamma_{\mathcal{B}}$, $\nu_{\mathcal{B}} = - \nu_{\mathcal{C}}$ on $\partial \Gamma_{\mathcal{B}} \cap \partial \Gamma_{\mathcal{C}}$, and $\nu_{\mathcal{C}} = - \nu_{\mathcal{A}}$ on $\partial \Gamma_{\mathcal{C}} \cap \partial \Gamma_{\mathcal{A}}$, to see that
\begin{multline*}
\int_{\partial \Gamma_{\mathcal{A}} (t)} \nu_{\mathcal{A}} \cdot \varphi { \ }d \mathcal{H}^1_x + \int_{\partial \Gamma_{\mathcal{B}} (t)} \nu_{\mathcal{B}} \cdot \varphi { \ }d \mathcal{H}^1_x + \int_{\partial \Gamma_{\mathcal{C}} (t)} \nu_{\mathcal{C}} \cdot \varphi { \ }d \mathcal{H}^1_x \\
= \int_{\partial \Gamma_{\mathcal{A}} \setminus ( \partial \Gamma_{\mathcal{B}} \cup \partial \Gamma_{\mathcal{C}} )} \nu_{\mathcal{A}} \cdot \varphi { \ }d \mathcal{H}^1_x + \int_{\partial \Gamma_{\mathcal{B}} \setminus ( \partial \Gamma_{\mathcal{A}}  \cup \partial \Gamma_{\mathcal{C}} )} \nu_{\mathcal{B}} \cdot \varphi { \ }d \mathcal{H}^1_x \\
+ \int_{\partial \Gamma_{\mathcal{C}} \setminus ( \partial \Gamma_{\mathcal{A}} \cup \partial \Gamma_{\mathcal{B}} )} \nu_{\mathcal{C}} \cdot \varphi { \ }d \mathcal{H}^1_x = \int_{\partial \Gamma (t)} \nu \cdot \varphi {  \ }d \mathcal{H}^1_x.
\end{multline*}
Therefore Theorem \ref{thm24} is proved. 
\end{proof}

\section{On Generalized Diffusion and Heat Systems on Evolving Surface}\label{sect4}

Let us derive and study the diffusion system \eqref{eq11} and the heat system \eqref{eq12} on the evolving surface $\Gamma (t)$ with a boundary. We first calculate a variation of our dissipation energy to prove Theorem \ref{thm25}. Then we derive \eqref{eq11} and \eqref{eq12} from an energetic point of view, and investigate the boundary conditions for the two systems to prove Theorem \ref{thm27}.

\subsection{Variation of Dissipation Energy}\label{subsec41}

Let us prove Theorem \ref{thm25}.
\begin{proof}[Proof of Theorem \ref{thm25}]
Fix $\psi \in C_0^1 (\Gamma (t) )$. A direct calculation gives
\begin{equation*}
\frac{d}{d \varepsilon} \bigg|_{\varepsilon = 0} E_{D} [f + \varepsilon \psi ](t) = - \int_{\Gamma (t)} e'_{D} ( |{\rm{grad}}_\Gamma f |^2 ) {\rm{grad}}_\Gamma f \cdot {\rm{grad}}_\Gamma \psi { \ }d \mathcal{H}^2_x .
\end{equation*}
Applying Theorem \ref{thm23}, we see that
\begin{multline*}
- \int_{\Gamma (t)} e'_{D} ( | \nabla_\Gamma f|^2 )\nabla_\Gamma f \cdot \nabla_\Gamma \psi { \ }d \mathcal{H}^2_x\\
 = \int_{\Gamma (t)} {\rm{div}}_\Gamma \{ e'_{D} ( | \nabla_\Gamma f |^2 ) \nabla_\Gamma f \} \psi { \ }d \mathcal{H}^2_x
 - \int_{\partial \Gamma (t)} e'_{D} ( | \nabla_\Gamma f |^2 ) ( \nu \cdot \nabla_\Gamma f)  \psi { \ }d \mathcal{H}^1_x.
\end{multline*}
Note that $n \cdot \nabla_\Gamma f = 0$ on $\Gamma (t)$. Since $\psi|_{\partial \Gamma (t)} = 0$, Theorem \ref{thm25} is proved. 
\end{proof}

\subsection{On Generalized Diffusion System on Evolving Surface with Boundary}\label{subsec42}

Let us derive the generalized diffusion system \eqref{eq11} on the evolving surface $\Gamma (t)$ with a boundary from an energetic point of view. We assume that $C$ is a $C^2$-function on $\overline{\Gamma_T}$. Let $e_1$ be a $C^1$-function. Fix $0 \leq t < T$. Set
\begin{equation*}
E_{GD} [ C ] (t) = - \int_{\Gamma (t)} \frac{1}{2} e_1 ( |{\rm{grad}}_\Gamma C |^2 ) { \ }d \mathcal{H}^2_x.
\end{equation*}
From \cite{KS17}, \cite{K18}, \cite{K19}, we find that $E_{GD} [ C ] (t)$ is a candidate of the energy dissipation due to general diffusion. From Theorem \ref{thm25} we obtain the following force:
\begin{equation*}
\frac{\delta E_{GD} }{\delta C} = {\rm{div}}_\Gamma \{ e_1' ( |{\rm{grad}}_\Gamma C |^2 ) {\rm{grad}}_\Gamma C \}.
\end{equation*}
We assume that the change of rate of the concentration equals to the force derived from a variation of the energy dissipation due to general diffusion. That is, for every $\Omega (t) \subset \Gamma (t)$, assume that
\begin{equation*}
\frac{d}{d t} \int_{\Omega (t)} C { \ }d \mathcal{H}^2_x = \int_{\Omega (t)} \frac{\delta E_{GD} }{\delta C}  { \ }d \mathcal{H}^2_x.
\end{equation*}
Using Proposition \ref{prop26}, we have the following diffusion system:
\begin{equation}\label{eq41}
D_t^w C + ({\rm{div}}_\Gamma w ) C = {\rm{div}}_\Gamma \{ e_1' ( |{\rm{grad}}_\Gamma C |^2 ) {\rm{grad}}_\Gamma C \} \text{ on }\Gamma_T .
\end{equation}

Now we consider the conservation and energy laws of the system \eqref{eq41}. Applying Proposition \ref{prop26} and Theorem \ref{thm23}, we find that for $t_1 < t_2$,
\begin{equation*}
\int_{\Gamma (t_2)} C { \ } d \mathcal{H}^2_x = \int_{\Gamma (t_1)} C { \ } d \mathcal{H}^2_x + \int_{t_1}^{t_2} e_1' (| {\rm{grad}}_\Gamma C |^2 ) (\nu \cdot \nabla_\Gamma ) C { \ }d \mathcal{H}^1_x d \tau
\end{equation*}
and
\begin{multline*}
\int_{\Gamma (t_2)} \frac{1}{2} |C (x,t_2) |^2 { \ } d \mathcal{H}^2_x + \int_{t_1}^{t_2} \int_{\Gamma (\tau)} e_1' (| {\rm{grad}}_\Gamma C |^2 ) | {\rm{grad}}_\Gamma C |^2 { \ }d \mathcal{H}^2_x d \tau \\
 = \int_{\Gamma (t_1)} \frac{1}{2}  | C (x,t_1) |^2 { \ } d \mathcal{H}^2_x + \int_{t_1}^{t_2} \int_{\partial \Gamma (\tau )} e_1' (| {\rm{grad}}_\Gamma C |^2 ) C ( \nu \cdot \nabla_\Gamma ) C { \ }d \mathcal{H}^1_x d \tau.
\end{multline*}
Therefore the assertions $(\mathrm{i})$ and $(\mathrm{ii})$ of Theorem \ref{thm27} is proved. It is easy to check that \eqref{eq41} is equivalent to \eqref{eq15}.

\subsection{On Generalized Heat System on Evolving Surface with Boundary}\label{subsec43}

Let us derive the generalized heat system \eqref{eq12} on the evolving surface $\Gamma (t)$ with a boundary from an energetic point of view. We assume that $\rho$ and $\theta$ are two $C^2$-functions on $\overline{\Gamma_T}$. Let $e_2$ be a $C^1$-function. Fix $0 \leq t < T$. Set
\begin{equation*}
E_{ TD} [ \theta ] (t) = - \int_{\Gamma (t)} \frac{1}{2} e_2 ( |{\rm{grad}}_\Gamma \theta |^2 ) { \ }d \mathcal{H}^2_x.
\end{equation*}
From \cite{KS17}, \cite{K18}, \cite{K19}, we find that $E_{TD} [ \theta ] (t)$ is a candidate of the energy dissipation due to thermal diffusion. From Theorem \ref{thm25} we obtain the following force:
\begin{equation*}
\frac{\delta E_{TD} }{\delta \theta} = {\rm{div}}_\Gamma \{ e_2' ( |{\rm{grad}}_\Gamma \theta |^2 ) {\rm{grad}}_\Gamma \theta  \}.
\end{equation*}
From Proposition \ref{prop26} we admit the following continuity equation:
\begin{equation}\label{eq42}
D_t^w \rho + ({\rm{div}}_\Gamma w ) \rho = 0 \text{ on } \Gamma_T.
\end{equation}
We assume that the change of rate of the heat energy equals to the force derived from a variation of the energy dissipation due to thermal diffusion. That is, for every $\Omega (t) \subset \Gamma (t)$, assume that
\begin{equation*}
\frac{d}{d t} \int_{\Omega (t)} \rho \theta { \ }d \mathcal{H}^2_x = \int_{\Omega (t)} \frac{\delta E_{TD} }{\delta \theta}  { \ }d \mathcal{H}^2_x.
\end{equation*}
Using Proposition \ref{prop26} and \eqref{eq42}, we have
\begin{equation*}
\rho D_t^w \theta = {\rm{div}}_\Gamma \{ e_2' ( |{\rm{grad}}_\Gamma \theta |^2 ) {\rm{grad}}_\Gamma \theta \} .
\end{equation*}
Thus, we have the following heat system:
\begin{equation}\label{eq43}
\begin{cases}
D_t^w \rho + ({\rm{div}}_\Gamma w ) \rho = 0 & \text{ on } \Gamma_T,\\
\rho D_t^w \theta = {\rm{div}}_\Gamma \{ e_2' (| {\rm{grad}}_\Gamma \theta |^2 ) {\rm{grad}}_\Gamma \theta \}  & \text{ on }\Gamma_T.
\end{cases}
\end{equation}

Now we consider the conservation and energy laws of the system \eqref{eq43}. Applying Proposition \ref{prop26}, Theorem \ref{thm23}, and \eqref{eq43}, we see that for $t_1 < t_2$,
\begin{equation*}
\int_{\Gamma (t_2)} \rho  \theta { \ } d \mathcal{H}^2_x = \int_{\Gamma (t_1)} \rho \theta { \ } d \mathcal{H}^2_x + \int_{t_1}^{t_2} e_2' (| {\rm{grad}}_\Gamma \theta |^2 ) (\nu \cdot \nabla_\Gamma ) \theta { \ }d \mathcal{H}^1_x d \tau
\end{equation*}
and
\begin{multline*}
\int_{\Gamma (t_2)} \frac{1}{2} \{ \rho | \theta |^2 \} (x,t_2) { \ } d \mathcal{H}^2_x + \int_{t_1}^{t_2} \int_{\Gamma (\tau)} e_2' (| {\rm{grad}}_\Gamma \theta |^2 ) | {\rm{grad}}_\Gamma \theta |^2 { \ }d \mathcal{H}^2_x d \tau \\ 
= \int_{\Gamma (t_1)} \frac{1}{2}  \{ \rho | \theta |^2 \} (x,t_1) { \ } d \mathcal{H}^2_x + \int_{t_1}^{t_2} \int_{\partial \Gamma (\tau )} e_2' (| {\rm{grad}}_\Gamma \theta |^2 ) \theta (\nu \cdot \nabla_\Gamma ) \theta { \ }d \mathcal{H}^1_x d \tau.
\end{multline*}
Therefore the assertions $(\mathrm{iii})$ and $(\mathrm{iv})$ of Theorem \ref{thm27} is proved. It is easy to check that \eqref{eq43} is equivalent to \eqref{eq16}.

\section{Divergence and Transport Theorems on Evolving Double Bubble}\label{sect5}

In this section we prepare the two key tools to make a mathematical model for a diffusion process on an evolving double bubble (see Fig. \ref{Fig1}). Let $\mathfrak{a} , \mathfrak{b} , \mathfrak{m} \in C^1([0,T))$ for some $T \in (0, \infty ]$. For every $0 \leq t < T$, assume that
\begin{equation*}
0 < \mathfrak{m} (t) < \mathfrak{a} (t) <\mathfrak{b} (t) < 2 \mathfrak{m} (t).
\end{equation*}
Set
\begin{equation*}
\mathfrak{n} = \mathfrak{n} (t) = \frac{\mathfrak{a}^2 - \mathfrak{b}^2}{4 \mathfrak{m} },
\end{equation*}
and
\begin{align*}
\Gamma_A (t) = \left\{ x \in \mathbb{R}^3; { \ }  \left( x_1 + \mathfrak{m} (t) \right)^2 + x_2^2 + x_3^2 = \mathfrak{a}^2 (t),{ \ } x_1 < \mathfrak{n} (t) \right\},\\
\Gamma_B (t) = \left\{ x \in \mathbb{R}^3; { \ }  \left( x_2 - \mathfrak{m} (t) \right)^2 + x_2^2 + x_3^2 = \mathfrak{b}^2 (t), { \ }x_1 > \mathfrak{n} (t) \right\},\\
\Gamma_S (t) = \left\{ x \in \mathbb{R}^3; { \ }   \left( x_1 + \mathfrak{m} (t) \right)^2 + x_2^2 + x_3^2 \leq \mathfrak{a}^2 (t),{ \ } x_1 = \mathfrak{n} (t) \right\},\\
\partial \Gamma_0 ( t) = \left\{ x \in \mathbb{R}^3; { \ }   \left( x_1 + \mathfrak{m} (t) \right)^2 + x_2^2 + x_3^2 = \mathfrak{a}^2 (t),{ \ } x_1 = \mathfrak{n} (t) \right\}.
\end{align*}
It is clear that
\begin{equation*}
\partial \Gamma_A (t) = \partial \Gamma_B (t) = \partial \Gamma_S (t) = \partial \Gamma_0(t) .
\end{equation*}
The aim of this section is to derive the divergence and transport theorems on the evolving double bubble. 
\begin{proposition}[Divergence theorem on double bubble]\label{prop51}{ \ }\\
For all $\varphi = { }^t ( \varphi_1 , \varphi_2 , \varphi_3 ) \in [C^1 ( \overline{\Gamma_\sharp (t)} )]^3$,
\begin{equation}\label{eq51}
\int_{\Gamma_\sharp (t)} {\rm{div}}_{\Gamma_\sharp} \varphi { \ }d \mathcal{H}^2_x = - \int_{\Gamma_\sharp (t)} H_{\Gamma_\sharp} ( n_\sharp \cdot \varphi  ) { \ } d \mathcal{H}^2_x + \int_{\partial \Gamma_0 (t)} \nu_\sharp \cdot \varphi { \ }d \mathcal{H}^1_x,
\end{equation}
where $\sharp = A, B ,S$. Here
\begin{equation*}
H_{\Gamma_A} = \frac{2}{\mathfrak{a}}, { \ \ \ } H_{\Gamma_B} = \frac{2}{\mathfrak{b}}, { \ \ \ }H_{\Gamma_S} = 0,
\end{equation*}
and
\begin{equation*}
\nu_A = 
\begin{pmatrix}
\frac{\sqrt{\mathfrak{a}^2 - (\mathfrak{m} + \mathfrak{n} )^2} }{\mathfrak{a}}\\
- \frac{\mathfrak{m} + \mathfrak{n} }{\mathfrak{a} } \cos \vartheta\\
- \frac{\mathfrak{m} + \mathfrak{n} }{\mathfrak{a} } \sin \vartheta
\end{pmatrix},{ \ }\nu_B = 
\begin{pmatrix}
- \frac{\sqrt{\mathfrak{b}^2 - (\mathfrak{m} + \mathfrak{n} )^2} }{\mathfrak{b}}\\
\frac{\mathfrak{m} + \mathfrak{n} }{\mathfrak{b} } \cos \vartheta\\
\frac{\mathfrak{m} + \mathfrak{n} }{\mathfrak{b} } \sin \vartheta
\end{pmatrix},{ \ }\nu_S = 
\begin{pmatrix}
0\\
\cos \vartheta\\
\sin \vartheta
\end{pmatrix}
\end{equation*}
at $x  = { }^t (\mathfrak{n} , \sqrt{ \mathfrak{a}^2 -(\mathfrak{m} + \mathfrak{n} )^2} \cos \vartheta , \sqrt{ \mathfrak{a}^2 -(\mathfrak{m} + \mathfrak{n} )^2}  \sin \vartheta )$, $\vartheta \in [0 , 2 \pi]$.
\end{proposition}

\begin{proposition}[Transport theorem on double bubble]\label{prop52}{ \ }\\
For all $f \in C ( \mathbb{R}^4)$ and $\Omega (t) \subset \Gamma_\sharp (t)$,
\begin{equation}\label{eq52}
\frac{d }{d t}\int_{\Omega (t)} f (x,t) { \ }d \mathcal{H}^2_x = \int_{\Omega (t)} \{ D_t^{w_\sharp} f + ({\rm{div}}_{\Gamma_\sharp} w_\sharp ) f \} (x,t) { \ }d \mathcal{H}^2_x .
\end{equation}
where $\sharp = A, B ,S$ and $w_\sharp$ is the motion velocity of $\Gamma_\sharp (t)$ given by \eqref{eq53}, \eqref{eq54}, \eqref{eq55}.
\end{proposition}

We first study the motion velocities of $\Gamma_A$, $\Gamma_B$, and $\Gamma_S$. To this end we divide the double bubble into some evolving surfaces with boundaries (see Fig. \ref{Fig2}). We set
\begin{align*}
\Gamma_{A1} (t) & = \left\{ x \in \Gamma_A (t);{ \ } -( \mathfrak{a} (t) + \mathfrak{m} (t) ) \leq x_1 < - \left( \frac{\mathfrak{a}(t)}{2} + \mathfrak{m} (t) \right) \right\},\\
\Gamma_{A2} (t) & = \left\{ x \in \Gamma_A (t);{ \ } - \left( \frac{\mathfrak{a} (t)}{2} + \mathfrak{m} (t) \right) < x_1 < \mathfrak{n} (t) \right\},\\
\Gamma_{B1} (t) & = \left\{ x \in \Gamma_B (t);{ \ } \mathfrak{n} (t)<  x_1 < \left( \frac{\mathfrak{b} (t)}{2} + \mathfrak{m} (t) \right) \right\},\\
\Gamma_{B2} (t) & = \left\{ x \in \Gamma_B (t);{ \ } \left( \frac{\mathfrak{b} (t)}{2} + \mathfrak{m} (t) \right) < x_1 \leq \mathfrak{b} (t) + \mathfrak{m} (t) \right\}.
\end{align*}
It is clear that
\begin{equation*}
\Gamma_A (t) = \overset{\circ}{\overline{ \Gamma_{A1} (t) \cup \Gamma_{A2} (t)} }, { \ }\Gamma_B (t) = \overset{\circ}{\overline{ \Gamma_{B1} (t) \cup \Gamma_{B2} (t)} }. 
\end{equation*}
For $r , \vartheta, t \in \mathbb{R}$,
\begin{equation*}
\widehat{x}_{A1} = \widehat{x}_{A1} (r , \vartheta ,t ) := \begin{pmatrix}
- \sqrt{ \mathfrak{a}^2 - \frac{3 \mathfrak{a}^2 r^2}{4} } -\mathfrak{m} \\
\frac{ \sqrt{3} }{2} \mathfrak{a} r \cos \vartheta\\
\frac{ \sqrt{3} }{2} \mathfrak{a} r  \sin \vartheta
\end{pmatrix},
\end{equation*}

\begin{equation*}
\widehat{x}_{A2} = \widehat{x}_{A2} (r , \vartheta ,t ) := \begin{pmatrix}
( \mathfrak{m}r + \mathfrak{n}r  + \mathfrak{a}r/2 - \mathfrak{a}/2) - \mathfrak{m} \\
\sqrt{ \mathfrak{a}^2 - (\mathfrak{m} r + \mathfrak{n} r + {\mathfrak{a} r}/{2}  - {\mathfrak{a}}/{2} )^2} \cos \vartheta\\
\sqrt{ \mathfrak{a}^2 - (\mathfrak{m} r + \mathfrak{n} r + {\mathfrak{a} r}/{2}  - {\mathfrak{a}}/{2} )^2} \sin \vartheta
\end{pmatrix},
\end{equation*}

\begin{equation*}
\widehat{x}_{B1} = \widehat{x}_{B1} (r , \vartheta ,t ) := \begin{pmatrix}
\mathfrak{m} - ( \mathfrak{m}r + \mathfrak{n}r  + \mathfrak{b}r/2 - \mathfrak{b}/2) \\
\sqrt{ \mathfrak{b}^2 - (\mathfrak{m} r + \mathfrak{n} r + {\mathfrak{b} r}/{2}  - {\mathfrak{b}}/{2} )^2} \cos \vartheta\\
\sqrt{ \mathfrak{b}^2 - (\mathfrak{m} r + \mathfrak{n} r + {\mathfrak{b} r}/{2}  - {\mathfrak{b}}/{2} )^2} \sin \vartheta
\end{pmatrix},
\end{equation*}

\begin{equation*}
\widehat{x}_{B2} = \widehat{x}_{B2} (r , \vartheta ,t ) := \begin{pmatrix}
\sqrt{ \mathfrak{b}^2 - \frac{3 \mathfrak{b}^2 r^2}{4} } + \mathfrak{m} \\
\frac{ \sqrt{3} }{2} \mathfrak{b} r \cos \vartheta\\
\frac{ \sqrt{3} }{2} \mathfrak{b} r \sin \vartheta
\end{pmatrix},
\end{equation*}

\begin{equation*}
\widehat{x}_S = \widehat{x}_S (r , \vartheta ,t ) := \begin{pmatrix}
\mathfrak{n}\\
\sqrt{\mathfrak{a}^2 - (\mathfrak{m} + \mathfrak{n} )^2} r \cos \vartheta\\
\sqrt{\mathfrak{a}^2 - (\mathfrak{m} + \mathfrak{n} )^2} r \sin \vartheta
\end{pmatrix}.
\end{equation*}
Set
\begin{align*}
U_{A1} & = U_{B2} =U_S = \left\{ { }^t (r , \vartheta ) \in \mathbb{R}^2; { \ }0 \leq r < 1,{ \ } 0 \leq \vartheta \leq 2 \pi  \right\},\\
U_{A2} & = U_{B1}= \left\{ { }^t (r , \vartheta ) \in \mathbb{R}^2; { \ }0 < r < 1,{ \ } 0 \leq \vartheta \leq 2 \pi  \right\}.
\end{align*}
Note that $U_{A1}$, $U_{A2}$, $U_{B1}$, $U_{B2}$, and $U_{S}$ do not depend on $t$. It is clear that
\begin{align*}
\Gamma_{A1} (t) & = \{ x \in \mathbb{R}^3; { \ }x = \widehat{x}_{A1} (r , \vartheta , t) , { \ }{ }^t( r , \vartheta ) \in U_{A1} \},\\
\Gamma_{A2} (t) & = \{ x \in \mathbb{R}^3; { \ }x = \widehat{x}_{A2} (r , \vartheta , t) , { \ }{ }^t( r , \vartheta ) \in U_{A2} \},\\
\Gamma_{B1} (t) & = \{ x \in \mathbb{R}^3; { \ }x = \widehat{x}_{B1} (r , \vartheta , t) , { \ }{ }^t( r , \vartheta ) \in U_{B1} \},\\
\Gamma_{B2} (t) & = \{ x \in \mathbb{R}^3; { \ }x = \widehat{x}_{B2} (r , \vartheta , t) , { \ }{ }^t( r , \vartheta ) \in U_{B2} \},\\
\Gamma_{S} (t)  & = \{ x \in \mathbb{R}^3; { \ }x = \widehat{x}_S (r , \vartheta , t) , { \ }{ }^t( r , \vartheta ) \in U_S \}.
\end{align*}
For $0 < t <T$ and $x \in \Gamma_A (t)$, set
\begin{equation}\label{eq53}
w_A (x, t) = \begin{cases}
 \frac{\partial \widehat{x}_{A1} }{\partial t} (r , \vartheta , t) \text{ at } x = \widehat{x}_{A1} (r, \vartheta , t) { \ }{ }^t (r, \vartheta ) \in U_{A1},\\
\frac{\partial \widehat{x}_{A2} }{\partial t} (r , \vartheta , t) \text{ at } x = \widehat{x}_{A2} (r, \vartheta , t) { \ }{ }^t (r, \vartheta ) \in \overline{U_{A2}}.
\end{cases}
\end{equation}
We call $w_A$ the motion velocity of $\Gamma_A(t)$. For $0 < t <T$ and $x \in \Gamma_B (t)$, set
\begin{equation}\label{eq54}
w_B (x, t) = \begin{cases}
 \frac{\partial \widehat{x}_{B1} }{\partial t} (r , \vartheta , t) \text{ at } x = \widehat{x}_{B1} (r, \vartheta , t), { \ }{ }^t (r, \vartheta ) \in \overline{U_{B1}},\\
\frac{\partial \widehat{x}_{B2} }{\partial t} (r , \vartheta , t) \text{ at } x = \widehat{x}_{B2} (r, \vartheta , t), { \ }{ }^t (r, \vartheta ) \in U_{B2}.
\end{cases}
\end{equation}
We call $w_B$ the motion velocity of $\Gamma_B(t)$. For $0 < t <T$ and $x \in \Gamma_S (t)$, set
\begin{equation}\label{eq55}
w_S (x, t) = \frac{\partial \widehat{x}_{S} }{\partial t} (r , \vartheta , t) \text{ at } x = \widehat{x}_S (r, \vartheta , t), { \ }{ }^t (r, \vartheta ) \in \overline{U_S}.
\end{equation}
We call $w_S$ the motion velocity of $\Gamma_S(t)$.

Now we study the divergence theorem on our evolving double bubble.
\begin{proof}[Proof of Proposition \ref{prop51}]
We first consider the case when $\sharp = A$. Set
\begin{align*}
\Gamma_{A11}(t) & = \left\{ x \in \mathbb{R}^3; { \ }x = \widehat{x}_{A1} (r , \vartheta , t), { \ }0 < r < 1,{ \ } 0 < \vartheta < \pi  \right\},\\
\Gamma_{A12}(t) &= \left\{ x \in \mathbb{R}^3; { \ }x = \widehat{x}_{A1} (r , \vartheta , t), { \ }0 < r < 1,{ \ } \pi < \vartheta < 2 \pi  \right\},\\
\Gamma_{A21}(t) &= \left\{ x \in \mathbb{R}^3; { \ }x = \widehat{x}_{A2} (r , \vartheta , t), { \ }0 < r < 1,{ \ } 0 < \vartheta < \pi  \right\},\\
\Gamma_{A22}(t) &= \left\{ x \in \mathbb{R}^3; { \ }x = \widehat{x}_{A2} (r , \vartheta , t ), { \ }0 < r < 1,{ \ } \pi < \vartheta < 2 \pi  \right\}.
\end{align*}
It is clear that
\begin{equation*}
\Gamma_{A1} (t) = \overset{\circ}{\overline{ \Gamma_{A11} (t) \cup \Gamma_{A12} (t)} } \text{ and }\Gamma_{A2} (t) = \overset{\circ}{\overline{ \Gamma_{A21} (t) \cup \Gamma_{A22} (t)} } .
\end{equation*}
Set $n^{U_{A1}} = { }^t (n_1^{U_{A1}}, n_2^{U_{A1}} ) = { }^t (1,0)$, $n^{U_{A21}} = { }^t (n_1^{U_{A21}}, n_2^{U_{A21}} ) = { }^t (1,0)$, $n^{U_{A22}} = { }^t (n_1^{U_{A22}}, n_2^{U_{A22}} ) = { }^t (-1,0)$. It is easy to see that $n^{U_{A1}}$ is the unit outer normal vector to $\partial U_{A1}$ and that $n^{U_{A21}}$, $n^{U_{A22}}$ are the unit outer normal vectors to $\partial U_{A2}$. A direct calculation shows that
\begin{align*}
\nu_{A1} & :=\frac{  n_1^{U_{A1}} \frac{\partial \widehat{x}_{A1}}{\partial X_2} }{ \left|  n_1^{U_{A1}} \frac{\partial \widehat{x}_{A1}}{\partial X_2} \right|} \times \frac{ \frac{\partial \widehat{x}_{A1}}{\partial X_1} \times \frac{\partial \widehat{x}_{A1}}{\partial X_2} }{ \left| \frac{\partial \widehat{x}_{A1}}{\partial X_1} \times \frac{\partial \widehat{x}_{A1}}{\partial X_2} \right|} = 
\begin{pmatrix}
\frac{\sqrt{3}r}{2}\\
\frac{ \sqrt{\mathfrak{a}^2 - 3 \mathfrak{a} r^2/4} }{\mathfrak{a}}  \cos \vartheta\\
\frac{ \sqrt{\mathfrak{a}^2 - 3 \mathfrak{a} r^2/4} }{\mathfrak{a}}  \sin \vartheta
\end{pmatrix},\\
\nu_{A21} & := \frac{  n_1^{U_{A21}} \frac{\partial \widehat{x}_{A2}}{\partial X_2} }{ \left|  n_1^{U_{A21}} \frac{\partial \widehat{x}_{A2}}{\partial X_2} \right|} \times \frac{ \frac{\partial \widehat{x}_{A2}}{\partial X_1} \times \frac{\partial \widehat{x}_{A2}}{\partial X_2} }{ \left| \frac{\partial \widehat{x}_{A2}}{\partial X_1} \times \frac{\partial \widehat{x}_{A2}}{\partial X_2} \right|} = 
\begin{pmatrix}
\frac{ \sqrt{\mathfrak{a}^2 - ( \mathfrak{m} r  + \mathfrak{n} r + \mathfrak{a} r/2 - \mathfrak{a} /2 )^2} }{\mathfrak{a}}\\
- \frac{ \mathfrak{m} r + \mathfrak{n} r + \mathfrak{a} r/2  - \mathfrak{a} /2 }{\mathfrak{a}} \cos \vartheta\\
- \frac{ \mathfrak{m} r  + \mathfrak{n} r + \mathfrak{a} r/2 - \mathfrak{a} /2 }{\mathfrak{a}} \sin \vartheta
\end{pmatrix},\\
\nu_{A22} & := \frac{  n_1^{U_{A22}} \frac{\partial \widehat{x}_{A2}}{\partial X_2} }{ \left|  n_1^{U_{A22}} \frac{\partial \widehat{x}_{A2}}{\partial X_2} \right|} \times \frac{ \frac{\partial \widehat{x}_{A2}}{\partial X_1} \times \frac{\partial \widehat{x}_{A2}}{\partial X_2} }{ \left| \frac{\partial \widehat{x}_{A2}}{\partial X_1} \times \frac{\partial \widehat{x}_{A2}}{\partial X_2} \right|} = 
\begin{pmatrix}
- \frac{ \sqrt{\mathfrak{a}^2 - ( \mathfrak{m} r  + \mathfrak{n} r + \mathfrak{a} r/2 - \mathfrak{a} /2 )^2} }{\mathfrak{a}}\\
\frac{ \mathfrak{m} r + \mathfrak{n} r + \mathfrak{a} r/2  - \mathfrak{a} /2 }{\mathfrak{a}} \cos \vartheta\\
\frac{ \mathfrak{m} r  + \mathfrak{n} r + \mathfrak{a} r/2 - \mathfrak{a} /2 }{\mathfrak{a}} \sin \vartheta
\end{pmatrix}.
\end{align*}
This gives
\begin{multline*}
\nu_{A1} (r= 1 , \vartheta ) = \begin{pmatrix}
\frac{\sqrt{3}}{2}\\
\frac{1}{2} \cos \vartheta\\
\frac{1}{2} \sin \vartheta
\end{pmatrix},{ \ \ \ }
\nu_{A21} (r= 1 , \vartheta ) = \begin{pmatrix}
\frac{ \sqrt{ \mathfrak{a} - (\mathfrak{m} + \mathfrak{n} )^2 }}{\mathfrak{a}}\\
\frac{\mathfrak{m} + \mathfrak{n} }{\mathfrak{a}} \cos \vartheta\\
\frac{\mathfrak{m} + \mathfrak{n} }{\mathfrak{a}} \sin \vartheta
\end{pmatrix},\\
\nu_{A22} (r= 0 , \vartheta ) = \begin{pmatrix}
- \frac{\sqrt{3}}{2}\\
- \frac{1}{2} \cos \vartheta\\
- \frac{1}{2} \sin \vartheta
\end{pmatrix}, \text{ and } \nu_{A1} (r=1 , \vartheta )= - \nu_{A22} (r = 0 , \vartheta ).
\end{multline*}
Applying Theorems \ref{thm23} and \ref{thm24}, we see that for all $\varphi \in C^1 ( \overline{\Gamma_A (t)})]^3$,
\begin{equation*}
\int_{\Gamma_A(t)} {\rm{div}}_{\Gamma_A} \varphi { \ }d \mathcal{H}^2_x = - \int_{\Gamma_A(t)} H_{\Gamma_A} ( n_A \cdot \varphi  ) { \ } d \mathcal{H}^2_x + \int_{\partial \Gamma_0 (t)} \nu_A \cdot \varphi { \ }d \mathcal{H}^1_x.
\end{equation*}
Here $\nu_A = \nu_{A21}$.

Next we consider \eqref{eq51} in the case when $\sharp = B$. Set
\begin{align*}
\Gamma_{B11}(t) & = \left\{ x \in \mathbb{R}^3; { \ }x = \widehat{x}_{B1} (r , \vartheta , t), { \ }0 < r < 1,{ \ } 0 < \vartheta < \pi  \right\},\\
\Gamma_{B12}(t) &= \left\{ x \in \mathbb{R}^3; { \ }x = \widehat{x}_{B1} (r , \vartheta , t ), { \ }0 < r < 1,{ \ } \pi < \vartheta < 2 \pi  \right\},\\
\Gamma_{B21}(t) &= \left\{ x \in \mathbb{R}^3; { \ }x = \widehat{x}_{B2} (r , \vartheta , t), { \ }0 < r < 1,{ \ } 0 < \vartheta < \pi  \right\},\\
\Gamma_{B22}(t) &= \left\{ x \in \mathbb{R}^3; { \ }x = \widehat{x}_{B2} (r , \vartheta , t), { \ }0 < r < 1,{ \ } \pi < \vartheta < 2 \pi  \right\}.
\end{align*}
It is clear that
\begin{equation*}
\Gamma_{B1} (t) = \overset{\circ}{\overline{ \Gamma_{B11} (t) \cup \Gamma_{B12} (t)} } \text{ and }\Gamma_{B2} (t) = \overset{\circ}{\overline{ \Gamma_{B21} (t) \cup \Gamma_{B22} (t)} } .
\end{equation*}
Applying Theorems \ref{thm23} and \ref{thm24}, we see that for all $\varphi \in C^1 ( \overline{\Gamma_B (t)})]^3$,
\begin{equation*}
\int_{\Gamma_B(t)} {\rm{div}}_{\Gamma_B} \varphi { \ }d \mathcal{H}^2_x = - \int_{\Gamma_B (t)} H_{\Gamma_B} ( n_B \cdot \varphi  ) { \ } d \mathcal{H}^2_x + \int_{\partial \Gamma_0 (t)} \nu_B \cdot \varphi { \ }d \mathcal{H}^1_x.
\end{equation*}
Note that 
\begin{equation*}
\nu_{B} = \frac{  - \frac{\partial \widehat{x}_{B2}}{\partial X_2} }{ \left|  - \frac{\partial \widehat{x}_{B2}}{\partial X_2} \right|} \times \frac{ \frac{\partial \widehat{x}_{B2}}{\partial X_1} \times \frac{\partial \widehat{x}_{B2}}{\partial X_2} }{ \left| \frac{\partial \widehat{x}_{B2}}{\partial X_1} \times \frac{\partial \widehat{x}_{B2}}{\partial X_2} \right|} = 
\begin{pmatrix}
- \frac{\sqrt{\mathfrak{b}^2 - (\mathfrak{m} + \mathfrak{n} )^2} }{\mathfrak{b}}\\
\frac{\mathfrak{m} + \mathfrak{n} }{\mathfrak{b} } \cos \vartheta\\
\frac{\mathfrak{m} + \mathfrak{n} }{\mathfrak{b} } \sin \vartheta
\end{pmatrix}{ \ }(r=1).
\end{equation*}

Finally, we consider \eqref{eq51} in the case when $\sharp = S$. Set
\begin{align*}
\Gamma_{S1}(t) & = \left\{ x \in \mathbb{R}^3; { \ }x = \widehat{x}_S (r , \vartheta , t), { \ }0 < r <1,{ \ } 0 < \vartheta < \pi  \right\},\\
\Gamma_{S2}(t) & = \left\{ x \in \mathbb{R}^3; { \ }x = \widehat{x}_S (r , \vartheta , t), { \ }0 < r < 1,{ \ } \pi < \vartheta < 2 \pi  \right\}.
\end{align*}
It is clear that
\begin{equation*}
\Gamma_{S1} (t) = \overset{\circ}{\overline{ \Gamma_{S1} (t) \cup \Gamma_{S2} (t)} }.
\end{equation*}
Applying Theorems \ref{thm23} and \ref{thm24}, we see that for all $\varphi \in C^1 ( \overline{\Gamma_S (t)})]^3$,
\begin{equation*}
\int_{\Gamma_S(t)} {\rm{div}}_{\Gamma_S} \varphi { \ }d \mathcal{H}^2_x = - \int_{\Gamma_S(t)} H_{\Gamma_S} ( n_S \cdot \varphi  ) { \ } d \mathcal{H}^2_x + \int_{\partial \Gamma_0 (t)} \nu_S \cdot \varphi { \ }d \mathcal{H}^1_x.
\end{equation*}
Note that 
\begin{equation*}
\nu_{S} = \frac{ \frac{\partial \widehat{x}_S}{\partial X_2} }{ \left| \frac{\partial \widehat{x}_S }{\partial X_2} \right|} \times \frac{ \frac{\partial \widehat{x}_S}{\partial X_1} \times \frac{\partial \widehat{x}_S }{\partial X_2} }{ \left| \frac{\partial \widehat{x}_S }{\partial X_1} \times \frac{\partial \widehat{x}_S }{\partial X_2} \right|} = 
\begin{pmatrix}
0\\
\cos \vartheta\\
\sin \vartheta
\end{pmatrix}.
\end{equation*}
Therefore Proposition \ref{prop51} is proved. 
\end{proof}

Finally, we study the transport theorem on our evolving double bubble.
\begin{proof}[Proof of Proposition \ref{prop52}]
We first consider the case when $\sharp = A$. Set
\begin{align*}
G_{A1} = \left( \frac{\partial \widehat{x}_{A1} }{\partial r} \cdot \frac{\partial \widehat{x}_{A1} }{\partial r} \right) \left( \frac{\partial \widehat{x}_{A1} }{\partial \vartheta} \cdot \frac{\partial \widehat{x}_{A1} }{\partial \vartheta} \right) - \left( \frac{\partial \widehat{x}_{A1} }{\partial r} \cdot \frac{\partial \widehat{x}_{A1} }{\partial \vartheta} \right)^2,\\
G_{A2} = \left( \frac{\partial \widehat{x}_{A2} }{\partial r} \cdot \frac{\partial \widehat{x}_{A2} }{\partial r} \right) \left( \frac{\partial \widehat{x}_{A2} }{\partial \vartheta} \cdot \frac{\partial \widehat{x}_{A2} }{\partial \vartheta} \right) - \left( \frac{\partial \widehat{x}_{A2} }{\partial r} \cdot \frac{\partial \widehat{x}_{A2} }{\partial \vartheta} \right)^2.
\end{align*}
Fix $\Omega_A (t) \subset \Gamma_A (t)$. By definition, there are $U_1 \subset U_{A1} $ and $U_2 \subset U_{A2}$ such that
\begin{multline*}
\int_{\Omega_A (t)} f (x,t) { \ } d \mathcal{H}^2_x = \int_{U_{1} \cap U_{A1}} f ( \widehat{x}_{A1} (r , \vartheta , t ), t) \sqrt{G_{A1}} { \ } d r d \vartheta\\
 +  \int_{U_2 \cap U_{A2}} f ( \widehat{x}_{A2} (r , \vartheta , t ) , t) \sqrt{G_{A2}} { \ } d r d \vartheta. 
\end{multline*}
Using \eqref{eq53} and \eqref{eq32}, we find that
\begin{equation*}
\frac{d}{d t}\int_{\Omega_A (t)} f { \ } d \mathcal{H}^2_x = \int_{\Omega_A (t)} D_t^{w_A } f + ({\rm{div}}_{\Gamma_A} w_A) f { \ }d \mathcal{H}^2_x. 
\end{equation*}
Similarly, we can derive \eqref{eq52} in the cases when $\sharp = B$ and $\sharp = S$. Therefore Proposition \ref{prop52} is proved. 
\end{proof}

\section{On Diffusion System on Evolving Double Bubble}\label{sect6}
In this section we make a mathematical model for a diffusion process on an evolving double bubble from an energetic point of view. Let us introduce our setting. Let $\Gamma_A(t)$, $\Gamma_B (t)$, and $\Gamma_S(t)$ be the three evolving surfaces defined by Section \ref{sect5}. Since $\partial \Gamma_A (t) = \partial \Gamma_B (t) = \partial \Gamma_S (t)$, we write $\partial \Gamma_0 (t)$ as $\partial \Gamma_\sharp (t)$. Set
\begin{multline*}
\Gamma_\sharp^T = \bigcup_{0 < t < T}\{ \Gamma_\sharp (t) \times \{ t \} \}, { \ }\partial \Gamma_\sharp^T = \bigcup_{0 < t < T}\{ \partial \Gamma_\sharp (t) \times \{ t \} \},\\
\overline{\Gamma_\sharp^T} = \bigcup_{0 \leq t < T}\{ \overline{\Gamma_\sharp (t)} \times \{ t \} \},{ \ } \frac{\partial C_\sharp}{\partial \nu_\sharp} = ( \nu_\sharp \cdot \nabla_{\Gamma_\sharp})C_\sharp,\\
 C^k (\overline{\Gamma_\sharp^T} ) = \{ f : \overline{\Gamma_\sharp^T} \to \mathbb{R}; { \ }f = \mathcal{F}|_{\overline{\Gamma_\sharp^T}}, \text{ for some }\mathcal{F} \in C^k (\mathbb{R}^4) \}.
\end{multline*}
Suppose that $C_\sharp \in C^2 ( \overline{\Gamma_\sharp^T} )$, $\kappa_\sharp \in C^1 ( \overline{\Gamma_\sharp^T} )$ and $\kappa_\sharp >0$, where $\sharp = A , B , S$. 
We assume that the concentrations $C_A$, $C_B$, and $C_S$ are equal to each other at the boundary $\partial \Gamma_0$. In other words, under the following restriction:
\begin{equation*}
C_A |_{\partial \Gamma_A (t)} = C_B |_{\partial \Gamma_B (t)} = C_S |_{\partial \Gamma_S (t)},
\end{equation*}
we make a mathematical model for a diffusion process on our evolving double bubble.

We first set the energy dissipation due to surface diffusion as follows:
\begin{equation*}
E_{SD} [ C_A , C_B , C_S ] = - \sum_{\sharp = A,B,S} \int_{\Gamma_\sharp (t)} \frac{\kappa_\sharp}{2} |{\rm{grad}}_{\Gamma_\sharp} C_\sharp |^2 { \ } d \mathcal{H}^2_x.
\end{equation*}
From \cite{KS17}, \cite{K18}, \cite{K19}, we find that $E_{SD}$ is a candidate of the energy dissipation due to surface diffusion. Applying Proposition \ref{prop51}, we check that for all $\psi_A  \in C^1 ( \overline{\Gamma_A(t)} )$, $\psi_B  \in C^1 ( \overline{\Gamma_B(t)} )$, and $\psi_S  \in C^1 ( \overline{\Gamma_S(t)} )$,
\begin{multline}\label{eq61}
\frac{d}{d \varepsilon} \bigg|_{\varepsilon = 0} E_{SD} [ C_A + \varepsilon \psi_A , C_B + \varepsilon \psi_B  , C_S + \varepsilon \psi_S  ] \\
= \sum_{\sharp = A,B,S}\int_{\Gamma_\sharp (t)} {\rm{div}}_{\Gamma_\sharp} \{ \kappa_\sharp \nabla_{\Gamma_\sharp} C_\sharp \} \psi_\sharp { \ } d \mathcal{H}^2_x - \sum_{\sharp = A,B,S} \int_{ \partial \Gamma_0} \kappa_\sharp \frac{\partial C_\sharp}{\partial \nu_\sharp} \psi_\sharp  { \ } d \mathcal{H}^1_x,
\end{multline}
where $- 1 < \varepsilon < 1$. Note that $n_\sharp \cdot ( \nabla_{\Gamma_\sharp} f) = 0$. Since
\begin{equation*}
C_A|_{\partial \Gamma_A} = C_B|_{\partial \Gamma_B} = C_S|_{\partial \Gamma_S},
\end{equation*}
we assume that for every $- 1 < \varepsilon < 1$,
\begin{equation*}
(C_A + \varepsilon \psi_A )|_{\partial \Gamma_A} = (C_B + \varepsilon \psi_B)|_{\partial \Gamma_B} = (C_S + \varepsilon \psi_S )|_{\partial \Gamma_S}.
\end{equation*}
Then we have
\begin{equation}\label{eq62}
\psi_A|_{\partial \Gamma_A} = \psi_B|_{\partial \Gamma_B} = \psi_S|_{\partial \Gamma_S}.
\end{equation}
Now we consider a variation of the dissipation energy $E_{SD}$. We assume that all $\psi_A  \in C^1 ( \overline{\Gamma_A(t)} )$, $\psi_B  \in C^1 ( \overline{\Gamma_B(t)} )$, and $\psi_S  \in C^1 ( \overline{\Gamma_S(t)} )$ satisfying \eqref{eq62},
\begin{equation*}
\frac{d}{d \varepsilon} \bigg|_{\varepsilon = 0} E_{SD} [ C_A + \varepsilon \psi_A , C_B + \varepsilon \psi_B  , C_S + \varepsilon \psi_S  ]= 0.
\end{equation*}
We first consider the case when $\psi_\sharp|_{\partial \Gamma_\sharp} = 0$. Since
\begin{equation*}
\frac{d}{d \varepsilon} \bigg|_{\varepsilon = 0} E_{SD} [ C_A + \varepsilon \psi_A , C_B + \varepsilon \psi_B  , C_S + \varepsilon \psi_S  ]= 0
\end{equation*}
for all $\psi_A  \in C_0^1 ( \Gamma_A(t) )$, $\psi_B  \in C_0^1 ( \Gamma_B(t) )$, and $\psi_S  \in C_0^1 ( \Gamma_S(t) )$, it follows from \eqref{eq61} to have
\begin{equation}\label{eq63}
\sum_{\sharp = A,B,S}1_{\Gamma_\sharp}{\rm{div}}_{\Gamma_\sharp} \{ \kappa_\sharp \nabla_{\Gamma_\sharp} C_\sharp \} = 0.
\end{equation}
Here
\begin{equation*}
1_{\Gamma_\sharp} = 1_{\Gamma_\sharp} (x)= \begin{cases}
1 \text{ if } x \in \Gamma_\sharp,\\
0 \text{ if } x \notin \Gamma_\sharp.
\end{cases}
\end{equation*}
Now we consider the case when $\psi_\sharp|_{\partial \Gamma_\sharp} \neq 0$. By \eqref{eq61} and \eqref{eq63}, we see that for all $\psi_A  \in C^1 ( \overline{\Gamma_A(t)} )$, $\psi_B  \in C^1 ( \overline{\Gamma_B(t)} )$, and $\psi_S  \in C^1 ( \overline{\Gamma_S(t)} )$, 
\begin{multline*}
\frac{d}{d \varepsilon} \bigg|_{\varepsilon = 0} E_{SD} [ C_A + \varepsilon \psi_A , C_B + \varepsilon \psi_B  , C_S + \varepsilon \psi_S  ] \\
= - \int_{ \partial \Gamma_0 (t)} \left( \kappa_A \frac{\partial C_A}{\partial \nu_A}\bigg|_{\partial \Gamma_A} + \kappa_B \frac{\partial C_B}{\partial \nu_B} \bigg|_{\partial \Gamma_B} + \kappa_S \frac{\partial C_S}{\partial \nu_S}\bigg|_{\partial \Gamma_S}  \right) \psi_0  { \ } d \mathcal{H}^1_x = 0,
\end{multline*}
where $\psi_0 = \psi_\sharp |_{\partial \Gamma_0}$. Since $\psi_0$ is arbitrary, we obtain
\begin{equation*}
\kappa_A \frac{\partial C_A}{\partial \nu_A}\bigg|_{\partial \Gamma_A} + \kappa_B \frac{\partial C_B}{\partial \nu_B} \bigg|_{\partial \Gamma_B} + \kappa_S \frac{\partial C_S}{\partial \nu_S}\bigg|_{\partial \Gamma_S} = 0.
\end{equation*}
Therefore we have
\begin{align*}
& \frac{\delta E_{SD}}{\delta C} = \sum_{\sharp = A,B,S}1_{\Gamma_\sharp}{\rm{div}}_{\Gamma_\sharp} \{ \kappa_\sharp \nabla_{\Gamma_\sharp} C_\sharp \},\\
& \kappa_A \frac{\partial C_A}{\partial \nu_A}\bigg|_{\partial \Gamma_A} + \kappa_B \frac{\partial C_B}{\partial \nu_B} \bigg|_{\partial \Gamma_B} + \kappa_S \frac{\partial C_S}{\partial \nu_S}\bigg|_{\partial \Gamma_S} = 0.
\end{align*}

We assume that the change of the concentration equals to the force derived from a variation of the dissipation energy $E_{SD}$. That is, for ever ball $B_r \subset \mathbb{R}^3$, assume that
\begin{multline*}
\frac{d}{d t} \left( \int_{\Gamma_A (t) \cap B_r} C_A { \ }d \mathcal{H}^2_x +  \int_{\Gamma_B (t) \cap B_r} C_B { \ }d \mathcal{H}^2_x + \int_{\Gamma_S (t) \cap B_r} C_S { \ }d \mathcal{H}^2_x \right) \\
= \int_{\Gamma_A (t) \cap B_r} \frac{\delta E_{SD}}{\delta C}  { \ }d \mathcal{H}^2_x +  \int_{\Gamma_B (t) \cap B_r} \frac{\delta E_{SD}}{\delta C}  { \ }d \mathcal{H}^2_x + \int_{\Gamma_S (t) \cap B_r} \frac{\delta E_{SD}}{\delta C}  { \ }d \mathcal{H}^2_x  .
\end{multline*}
Then we use Proposition \ref{prop52} to have
\begin{equation}\label{eq64}
\begin{cases}
D_t^{w_A} C_A + ({\rm{div}}_{\Gamma_A} w_A ) C_A ={\rm{div}}_{\Gamma_A} \{ \kappa_A {\rm{grad}}_{\Gamma_A} C_A \} & \text{ on }\Gamma_A^T,\\
D_t^{w_B} C_B + ({\rm{div}}_{\Gamma_B} w_B ) C_B ={\rm{div}}_{\Gamma_B} \{ \kappa_B {\rm{grad}}_{\Gamma_B} C_B \} & \text{ on } \Gamma_B^T, \\
D_t^{w_S} C_S + ({\rm{div}}_{\Gamma_S} w_S ) C_S ={\rm{div}}_{\Gamma_S} \{ \kappa_S {\rm{grad}}_{\Gamma_S} C_S \} & \text{ on } \Gamma_S^T.
\end{cases}
\end{equation}

Finally, we consider the conservation and energy laws of the system \eqref{eq64} under the following boundary conditions:
\begin{equation}\label{eq65}
\begin{cases}
C_A|_{\partial \Gamma_A} = C_B|_{\partial \Gamma_B} = C_S|_{\partial \Gamma_S},\\
\kappa_A \frac{\partial C_A}{\partial \nu_A}\bigg|_{\partial \Gamma_A} + \kappa_B \frac{\partial C_B}{\partial \nu_B} \bigg|_{\partial \Gamma_B} + \kappa_S \frac{\partial C_S}{\partial \nu_S}\bigg|_{\partial \Gamma_S} = 0.
\end{cases}
\end{equation}
Applying Propositions \ref{prop51} and \ref{prop52}, \eqref{eq64}, and \eqref{eq65}, we see that
\begin{multline*}
\frac{d}{d t} \left( \int_{\Gamma_A (t) } C_A { \ }d \mathcal{H}^2_x +  \int_{\Gamma_B (t) } C_B { \ }d \mathcal{H}^2_x + \int_{\Gamma_S (t)} C_S { \ }d \mathcal{H}^2_x \right) \\
= - \int_{\partial \Gamma_0 (t)} \left( \kappa_A \frac{\partial C_A}{\partial \nu_A}\bigg|_{\partial \Gamma_A} + \kappa_B \frac{\partial C_B}{\partial \nu_B} \bigg|_{\partial \Gamma_B} + \kappa_S \frac{\partial C_S}{\partial \nu_S}\bigg|_{\partial \Gamma_S} \right) { \ }d \mathcal{H}^1_x = 0
\end{multline*}
and that
\begin{multline*}
\frac{1}{2} \frac{d}{d t} \left( \int_{\Gamma_A (t) } |C_A|^2 { \ }d \mathcal{H}^2_x +  \int_{\Gamma_B (t) } |C_B|^2 { \ }d \mathcal{H}^2_x + \int_{\Gamma_S (t)} |C_S|^2 { \ }d \mathcal{H}^2_x \right) \\
= - \sum_{\sharp = A,B,S }\int_{\Gamma_\sharp ( t )} \kappa_\sharp |{\rm{grad}}_{\Gamma_\sharp} C_\sharp |^2 { \ } d \mathcal{H}^2_x + \sum_{\sharp = A,B,S }\int_{\partial \Gamma_0 ( t )} C_\sharp \kappa_\sharp (\nu_\sharp \cdot \nabla_{\Gamma_\sharp} )C_\sharp { \ } d \mathcal{H}^1_x\\
= - \sum_{\sharp = A,B,S }\int_{\Gamma_\sharp ( t )} \kappa_\sharp |{\rm{grad}}_{\Gamma_\sharp} C_\sharp |^2 { \ } d \mathcal{H}^2_x.
\end{multline*}
Integrating with respect to time, we have the conservation law \eqref{eq18} and the energy law \eqref{eq19}. It is easy to check that \eqref{eq64} is equivalent to \eqref{Eq110}.

\end{document}